%% file: dynamic_inf.tex
\documentclass[sigconf]{acmart}

\input{macros}

\setcopyright{rightsretained}

\acmConference[X]{X}{2018}{X}
\acmYear{2018}
\copyrightyear{2018}

\acmArticle{4}
\acmPrice{15.00}

\begin{document}
\title{Fast Influence Maximization in Dynamic Graphs: \\ A Local Updating Approach}
	
\author{Vijaya Krishna Yalavarthi}
\affiliation{%
	\institution{NTU Singapore}
}
\email{yalavarthi@ntu.edu.sg}

\author{Arijit Khan}
\affiliation{%
	\institution{NTU Singapore}
}
\email{arijit.khan@ntu.edu.sg}

\input{abstract}
	
\maketitle

\input{abstract}
\input{Introduction}	
\input{related}

\input{preliminaries}

\input{Prop_solution3}
\input{experiments}

\input{conclusion}
\vspace{-1mm}
{\scriptsize
\bibliographystyle{ACM-Reference-Format}
\bibliography{ref,ref_rel}
}
\vspace{-3mm}
\appendix
\input{appendix}

\end{document}

%% file: macros.tex
\usepackage{graphicx}
\usepackage{booktabs}
\usepackage{amsmath}
\usepackage{amsthm}
\usepackage{amssymb}
\usepackage{amsfonts}
\usepackage{epsfig}
\usepackage{graphicx}
\usepackage{subfigure}
\usepackage{hyperref}
\usepackage{breakurl}
\usepackage[bottom]{footmisc}
\usepackage{balance}
\usepackage{algpseudocode}
\usepackage{algorithmicx}
\usepackage{algorithm}
\usepackage{multirow}
\usepackage{xspace}
\usepackage[singlelinecheck=off]{caption}
\DeclareCaptionType{copyrightbox}
\usepackage{pifont}
\newtheorem{thrm}{Theorem}

\newtheorem{lma}{Lemma}

\newtheorem{defn}{Definition}

\newtheorem{exple}{Example}
\newtheorem{problem}{Problem}

\newcommand{\spara}[1]{\smallskip\noindent{\bf #1}}

\newcommand{\squishlist}{
 \begin{list}{$\bullet$}
  {  \setlength{\itemsep}{0pt}
     \setlength{\parsep}{3pt}
     \setlength{\topsep}{3pt}
     \setlength{\partopsep}{0pt}
     \setlength{\leftmargin}{2em}
     \setlength{\labelwidth}{1.5em}
     \setlength{\labelsep}{0.5em}
} }
\newcommand{\squishlisttight}{
 \begin{list}{$\bullet$}
  { \setlength{\itemsep}{0pt}
    \setlength{\parsep}{0pt}
    \setlength{\topsep}{0pt}
    \setlength{\partopsep}{0pt}
    \setlength{\leftmargin}{2em}
    \setlength{\labelwidth}{1.5em}
    \setlength{\labelsep}{0.5em}
} }

\newcommand{\squishdesc}{
 \begin{list}{}
  {  \setlength{\itemsep}{0pt}
     \setlength{\parsep}{3pt}
     \setlength{\topsep}{3pt}
     \setlength{\partopsep}{0pt}
     \setlength{\leftmargin}{1em}
     \setlength{\labelwidth}{1.5em}
     \setlength{\labelsep}{0.5em}
} }

\newcommand{\squishend}{
  \end{list}
}

\newcommand{\eat}[1]{}

\newcommand{\NP}{\ensuremath{\mathbf{NP}}\xspace}
\newcommand{\sharpP}{\ensuremath{\mathbf{\#P}}\xspace}

\newcounter{ccc}

\DeclareMathOperator*{\argmax}{arg\,max}
\newcommand{\bigO}{\mathcal{O}}

%% file: abstract.tex
\begin{abstract}
We propose a {\em generalized} framework for influence maximization in large-scale, time evolving networks.
Many real-life influence graphs such as social networks, telephone networks, and IP traffic data exhibit dynamic characteristics,
e.g., the underlying structure and communication patterns evolve with time. Correspondingly, we develop a
dynamic framework for the influence maximization problem, where we perform effective {\em local updates} to
quickly adjust the top-$k$ influencers, as the structure and communication patterns in the network change.
We design a novel {\sf N-Family} method (N=1, 2, 3, $\ldots$) based on the maximum influence arborescence (MIA)
propagation model with approximation guarantee of $(1-1/e)$. We then develop heuristic algorithms
by extending the {\sf N-Family} approach to other information propagation models (e.g., independent cascade, linear threshold)
and influence maximization algorithms (e.g., CELF, reverse reachable sketch).
Based on a detailed empirical analysis over several real-world, dynamic, and large-scale networks,
we find that our proposed solution, {\sf N-Family} improves the updating time of the top-$k$ influencers by $1\sim2$
orders of magnitude, compared to state-of-the-art algorithms, while ensuring similar memory usage and influence spreads.
\end{abstract} 

%% file: Introduction.tex
\vspace{-3mm}
\section{Introduction}
\label{sec:intr}
\vspace{-1mm}
The problem of influence analysis \cite{KKT03} has been widely studied in the
context of social networks, because of the tremendous number of applications of this
problem in viral marketing and recommendations.
The assumption
in bulk of the literature on this problem is that a
static network has already been provided, and the objective is to identify the
top-$k$ seed users in the network such that the expected number of
influenced users, starting from those seed users and following
an influence diffusion model, is maximized.

In recent years, however, people recognized
the inherent usefulness
in studying the dynamic network setting  \cite{AS14}, and
influence analysis is no exception to this general trend~\cite{SLCHT17,OAYK16},
because many real-world social networks
evolve over time. In an evolving graph, new edges (interactions)
and nodes (users) are continuously added,
while old edges and nodes get dormant, or deleted.
In addition, the communication pattern and frequency may also change.
\begin{figure}[t!]
	\centering
	\includegraphics[scale=0.34]{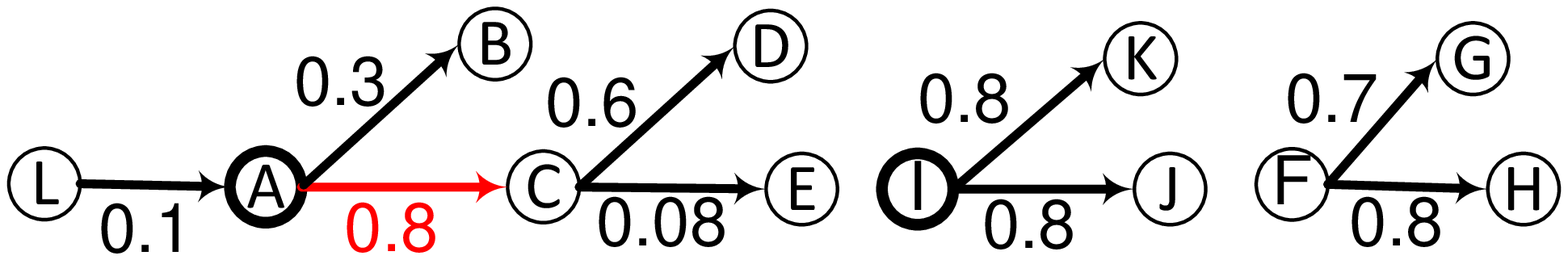}
	\vspace{-3mm}
	\caption{\small Running example: an influence graph}
	\label{fig:Inf_Prop}
	\vspace{-3mm}
\end{figure}
\begin{figure}[t!]
	\centering
	\includegraphics[scale=0.34]{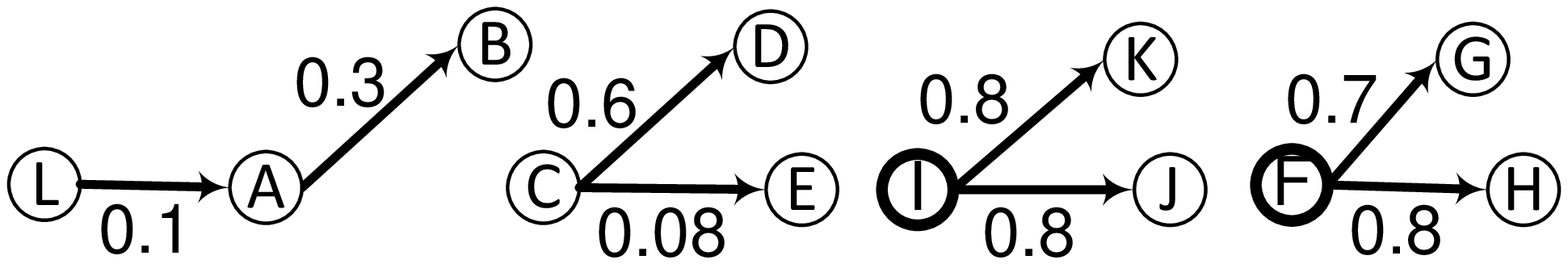}
	\vspace{-3mm}
	\caption{\small Influence graph after update operation: edge deletion $AC$}
	\label{fig:Inf_Prop_2}
	\vspace{-5mm}
\end{figure}
From an influence analysis perspective, even modest changes in the underlying
network structure (e.g., addition/ deletion of nodes and edges)
and communication patterns (e.g., update in influence probabilities over time)
may lead to changes in the top-$k$ influential nodes.
As an example, let us consider the influence graph in Figure~\ref{fig:Inf_Prop}
with 12 nodes, out of which the top-$2$ seed nodes are $A$ and $I$ (marked in bold),
following the Maximum Influence Arborescence (MIA) model and $\theta = 0.07$ \cite{CWW10} (we shall introduce
the details of the MIA model later). The influence spread obtained from this seed set, according to the MIA model, is: $2.58 + 2.6 = 5.18$.
Now, assume an update operation in the form of an edge removal $AC$ (marked in red).
The new influence spread obtained from the old seed nodes would be: $1.3 + 2.6 = 3.9$,
whereas if we recompute the top-$2$ seed nodes, they are $I$ and $F$, as shown in Figure~\ref{fig:Inf_Prop_2}.
The influence spread from these new seed nodes is: $2.6 + 2.5 = 5.1$.
It can be observed that there is a significant difference in the influence spread obtained with the old seed set vs. the new ones
(even for such a small example graph), which motivates us to efficiently update the seed nodes when the influence graph evolves.

However, computing the seed set from ground, after every update, is prohibitively expensive \cite{SLCHT17,OAYK16} ---
this inspires us to develop dynamic influence maximization algorithms.
By carefully observing, we realize that among the initial two seed nodes,
only one seed node, namely $A$ is replaced by $F$, whereas $I$ still continues to be a seed node.
It is because $A$ is in the {\em affected} region of the update operation, whereas $I$ is not affected by it.
Therefore, if we can identify that $A$ can no longer continue as a seed node, then we can remove it from the seed set; and
next, the aim would be to find one new seed node instead of two. Hence, we save almost $1/2$ of the computation in updating the seed set.

To this end, the two following questions are critical for identifying the top-$k$ seed nodes in a dynamic environment.
\vspace{-1mm}
\begin{itemize}
\setlength\itemsep{0.01em}
\item What are the regions affected when the network evolves?
\item How to efficiently update the seed nodes with respect to such affected regions?
\end{itemize}

\vspace{-2mm}
\spara{Affected region.} The foremost query that we address is identifying the affected region, i.e., the set of nodes potentially
affected due to the update. They could be: (1) the nodes (including some old seed nodes) whose influence spreads are {\em significantly}
changed due to the update operation, and also (2) those nodes whose {\em marginal gains} might change due to an affected seed node,
discovered in the previous step(s). Given a seed set $S$, the marginal gain of a node $v \not \in S$ is computed as the additional influence that $v$
can introduce when it is added to the seed set.

Given the influence graph and dynamic updates, we design an iterative algorithm to quickly identify the nodes in the affected region.
We call our method {\sf N-Family}, $N=1, 2, 3, \ldots,$ (until a base condition is satisfied), which we shall discuss in Section~\ref{sec:prop_sol}.

\vspace{-1mm}
\spara{Updating the seed nodes.} Once the affected region is identified, updating the top-$k$ seed set with respect to that affected region is also
a challenging problem. In this work, we develop an approximate algorithm under the MIA model of information diffusion, with {\em theoretical performance guarantee} of $1-1/e$.

Moreover, it should be understood that our primary aim is to maximize the influence spread as much as possible
with the new seed nodes, instead of searching for the exact seed nodes
(in fact, finding the exact seed nodes is \NP-hard \cite{KKT03}).
Therefore, we also show how to design more efficient heuristic algorithms, by carefully tuning
the parameters (e.g., by limiting $N=2$) of our {\sf N-Family} approach.

Our proposed framework to update the top-$k$ seed nodes is a {\em generic} one, and we develop heuristics by using it
on top of other information propagation models (e.g., independent cascade, linear threshold \cite{KKT03})
and several influence maximization (IM) algorithms (e.g., Greedy \cite{KKT03}, CELF \cite{LKGFVG07}, RR-sketch \cite{BBCL14,TSX15}).
In particular, we first find the affected region, and then update the seed nodes {\em only by adding a few sub-routines to the existing
static IM algorithms}, so that they can easily adapt to dynamic changes.

\vspace{-1mm}
\spara{Our contributions.} The contributions of our work can be summarized as follows.
\vspace{-1mm}
\begin{itemize}
\setlength\itemsep{0.01em}
\item We propose an iterative technique, {\sf N-Family} that systematically identifies affected nodes (including old seed nodes) due to
dynamic updates, and develop an incremental method that replaces the affected seed nodes with new ones, so to maximize the influence spread in the
updated graph. We derive theoretical performance guarantees of our algorithm under the MIA model.
\item We show how to develop efficient heuristics by extending proposed algorithm to other information propagation models and influence maximization
algorithms for updating the seed nodes in an evolving network.
\item We conduct a thorough experimental evaluation using several real-world, dynamic, and large graph datasets. The
empirical results with our heuristics attest $1\sim2$ orders of efficiency improvement, compared to state-of-the-art approaches \cite{SLCHT17,OAYK16}.
A snippet of our results is presented in Table~\ref{tab:result_summary}.
\end{itemize}
\begin{table}[tb!]
	\vspace{-2mm}
	\scriptsize
	\begin{center}
		\vspace{-1mm}
		\centering
		\caption{\small Average seed-set-updating time (sec) per node addition in the influence graph; the seed set consists of top-30 seed nodes;
			IC Model for influence cascade. For more details, we refer to Section~\ref{sec:experiments}\label{tab:result_summary}.}
		\vspace{-3mm}
		\begin{tabular} { c||c|c||c|c}
			\hline
			\textsf{Datasets}             & \textsf{UBI+} & \textsf{Family-CELF} & \textsf{DIA} & \textsf{Family-RRS} \\
			(\#nodes, \#edges)            & \cite{SLCHT17}& [our method]         & \cite{OAYK16}&[our method]   \\ \hline \hline
			{\em Digg} (30K, 85K)         &   3.36 sec    &  {\bf 0.008} sec      &    5.60 sec &    0.20 sec    \\ \hline
			{\em Slashdot} (51K, 130K)    &   11.3 sec    &  {\bf 0.05} sec      &   35.16 sec &  2.96   sec    \\ \hline
			{\em Epinions} (0.1M, 0.8M)   &   1111.21 sec   &  24.58 sec      &    134.68 sec &    {\bf 5.31} sec   \\ \hline
			{\em Flickr} (2.3M, 33M)      &   45108.09    sec   &  1939.40 sec          &    770.41 sec&    {\bf 273.50} sec   \\ \hline
		\end{tabular}
	\end{center}
	\vspace{-6mm}
\end{table} 

%% file: related.tex
\vspace{-2mm}
\section{Related Work}
\label{related}

\vspace{-1mm}
Kempe et al. \cite{KKT03} addressed the problem of
influence maximization in a social network as a discrete optimization problem,
and proposed a hill climbing greedy algorithm,
with an accuracy guarantee of $(1 -1/e)$.
They used the MC
simulation to compute the expected influence spread of a seed set.
Since the introduction of the influence maximization problem, many algorithms (see \cite{CLC13} for details)
have been developed, both heuristic and approximated, to improve the efficiency of the original greedy method.
Below, we survey some of these methods that are employed in our framework.
Leskovec et al. \cite{LKGFVG07} exploited the sub-modularity property of the greedy algorithm, and proposed more efficient {\sf CELF} algorithm.
Chen et al. \cite{CWW10} avoided MC simulations, and
developed the MIA model using maximum probable paths for the
influence spread computation. Addressing the inefficiency of MC simulations,
Borgs et al. \cite{BBCL14} introduced a reverse reachable sketching technique ({\sf RRS})
without sacrificing the accuracy guarantee.

In recent years, there has been interest in performing
influence analysis in dynamic graphs~\cite{ALY12,SLCHT17,LLLZSHX15, ZSTZS13,OAYK16,WFLT17}.
The work in~\cite{ALY12} was the first to propose methods that maximize the influence over a specific
interval in time; however, it was not designed for the online setting.
The work in~\cite{ZSTZS13} probed a subset of the nodes for detecting the underlying changes.
Liu et al. \cite{LLLZSHX15} considered an evolving network model (e.g., preferential attachment) for influence
maximization. Subbian et al.~\cite{SAS16b} discussed the problem of finding
influencers in social streams, although they employed frequent pattern mining techniques over
the underlying social stream of content. This is a different modeling assumption than the dynamic graph
setting considered in this work. Recently, Wang et al. \cite{WFLT17} considered a {\em sliding window
model} to find influencers based on the most recent interactions. Once again, their framework is {\em philosophically different
from the classical influence maximization setting \cite{KKT03}, as they do not consider any edge probabilities};
and hence, not directly comparable to ours.

In regards to problem formulation, recent works in \cite{SLCHT17,OAYK16,YWPC17} are closest to ours.
{\sf UBI+} \cite{SLCHT17} was designed for MC-simulation based algorithms and IC model. It performs
greedy exchange for multiple times --- every time an old seed node is replaced
with the best possible non-seed node. If one continues such exchanges until there is no improvement, the method guarantees
0.5 approximation. {\sf DIA} \cite{OAYK16} and \cite{YWPC17} work on top of RR-Sketches. These methods generate
all RR-sketches only once; and after every update, quickly modifies those existing sketches.
After that, {\em {\sf DIA} \cite{OAYK16} identifies all seed nodes from ground using modified sketches}.
This is the key difference with our framework, since {\em we generally need to identify
only a limited number of new seed nodes}, based on affected regions due to updates.
On the contrary, \cite{YWPC17} reports the top-$k$ nodes having maximum influence
spreads individually with the modified sketches. Thus, the objective of \cite{YWPC17} is
different from that of classic influence maximization, which we study in this work.

Moreover, it is non-trivial to adapt both {\sf UBI+} and {\sf DIA} for other
influence models and IM algorithms, than their respective ones. A drawback of this is as follows.
Sketch based methods (e.g., {\sf DIA}) consume higher memory for storing multiple sketches.
In contrast, MC-simulation based methods (e.g., {\sf UBI+}) are slower over large graphs.
On the other hand, our proposed {\sf N-Family} approach can
be employed over many IM models and algorithms, and {\em due to the local updating principle,
it significantly improves the efficiency under all scenarios}. Therefore, one can select the underlying
IM models and algorithms for the {\sf N-Family} approach based on system specifications and application
requirements. This demonstrates the {\em generality} of our solution.

%% file: preliminaries.tex
\vspace{-3mm}
\section{Preliminaries}
\label{sec:preliminaries}
\vspace{-1mm}
An influence network can be modeled as an uncertain graph $\mathcal{G}(V,E,$ $P)$, where $V$ and $E\subseteq V\times V$ denote the sets of nodes (users) and directed edges (links between users) in the network, respectively. $P$ is a function $P : E \rightarrow (0, 1)$ that assigns a probability to every edge $uv \in E$, such that $P_{uv}$ is the strength at which an active user $u \in V$ influences her neighbor $v \in V$. The edge probabilities can be learnt (from past propagation traces), or inferred (following various models), as discussed in \cite{CLC13}. In this work, we shall assume that $\mathcal{G}(V,E,P)$ is given as an input to our problem.
\begin{algorithm}[tb!]
	\caption{\small $Greedy(\mathcal{G},S,k)$: for IM in static networks}
    \label{alg:greedy}
	\begin{algorithmic}[1]
		\Require Graph $\mathcal{G}(V, E, P)$, seed set $S$ (initially empty), positive integer $k$
		\Ensure Seed set $S$ having the top-$k$ seed nodes
		\While {$|S| <= k$}
		\State $u^* = \argmax_{u \in V\setminus S} \{\sigma (S\cup \{u\}) - \sigma (S)\}$
		\State $S = S\cup u^*$
		\EndWhile
		\State Output $S$
	\end{algorithmic}	
\end{algorithm}
\vspace{-3mm}
\subsection{Influence Maximization in Static Graphs}
\label{sec:statIM}
\vspace{-1mm}
Whenever a social network user $u$ buys a product, or endorses an action (e.g., re-tweets a post),
she is viewed as being influenced or activated. When $u$ is active, she automatically becomes eligible to influence her
neighbors who are not active yet. While our designed framework can be applied on top of a varieties of influence diffusion models; due to brevity, we
shall introduce maximum influence arborescence (MIA) \cite{CWW10} and independent cascade (IC) \cite{KKT03} models.
We develop an approximate algorithm with theoretical guarantee on top of MIA, and an efficient heuristic with IC.

\spara{MIA model.} We start with an already active set of nodes $S$, called the seed set, and the influence from the seed nodes propagates only via the {\em maximum influence paths}.
A path $Pt$ from a source $u$ to a destination node $v$ is called the maximum influence path $MIP(u,v)$ if this has the highest probability compared
to all other paths between the same pair of nodes. Ties are broken in a predetermined and consistent way, such that the maximum
influence path between a pair of nodes is always unique. Formally,
\vspace{-1mm}
\begin{align}
MIP(u,v) = \argmax_{Pt \in \mathcal{P}(u,v)} \{\prod_{e\in Pt} P_e\} \label{eq:mip}
\vspace{-2mm}
\end{align}
Here, $\mathcal{P}(u,v)$ denotes the set of all paths from $u$ to $v$. In addition, an influence threshold $\theta$ (which is
an input parameter to trade off between efficiency and accuracy \cite{CWW10}) is used to eliminate
maximum influence paths that have smaller propagation probabilities than $\theta$.

\vspace{-1mm}
\spara{IC model.} This model assumes that diffusion process from the seed nodes continue in {\em discrete time steps}. When some node $u$ first becomes active at step
$t$, it gets a single chance to activate each of its currently inactive out-neighbors $v$; it succeeds with probability $P_{uv}$. If
$u$ succeeds, then $v$ will become active at step $t+1$. Whether or not $u$ succeeds at step $t$, it cannot make any further attempts in
the subsequent rounds. Each node can be activated only once and it stays active until the end. The campaigning process runs until no more activations are possible.

\spara{Influence estimation problem.} All active nodes at the end, due to a diffusion process, are considered as the nodes influenced by $S$.
In an uncertain graph $\mathcal{G}(V,E,P)$, influence estimation is the problem of identifying the expected influence spread $\sigma(S)$ of $S \subseteq V$.

It has been proved that the exact estimation of influence spread is a \sharpP-hard problem,
under the IC model \cite{CWW10}. However, influence spread can be computed in polynomial time for the MIA model.

\spara{Marginal influence gain.} Given a seed set $S$, the marginal gain $MG(S,u)$ of a node $u \not \in S$ is computed as the additional influence that $u$
can introduce when it is added to the seed set.
\begin{align}
	MG(S,u) = \sigma(S\cup\{u\})-\sigma(S) \label{eq:marginal}
\end{align}
\vspace{-1mm}
\spara{Influence maximization (IM) problem.}
Influence maximization is the problem of identifying the seed set $S^*$ of cardinality $k$ that has the maximum expected influence spread in the network.
\begin{table}[tb!]
	\vspace{-2mm}
	\begin{small}
		\begin{center}
			\vspace{-1mm}
			\centering
			\caption{\small Notations used and their meanings\label{tab:notation}}
			\vspace{-3mm}
			\begin{tabular} { |l|l|}
				\hline
				$\quad$ \textsf{Symbol} & $\qquad$ $\qquad$\textsf{Meaning} \\ \hline \hline
				$\mathcal{G}(V,E,P)$ & uncertain graph \\ \hline
				$P_e$ & probability of edge $e$ \\ \hline
				$Pt$ & a path \\ \hline
				$\mathcal{P}(u,v)$ & set of all paths from $u$ to $v$ \\ \hline
				$MIP(u,v)$ & the highest probability path from $u$ to $v$ \\ \hline
				$S$ & seed set \\ \hline
				$S^{i-1}$ & seed set formed after $(i-1)$ iterations of Greedy algorithm\\ \hline
				$s^i$ & seed node added at the $i$-th iteration of Greedy algorithm\\ \hline
				$\sigma(S)$ & expected influence spread from $S$ \\ \hline
				$pp(S,u)$ & probability that $u$ gets activated by $S$ \\ \hline
				$MG(S,u)$ & marginal influence gain of $u$ w.r.t. seed set $S$ \\ \hline
				$Q$ & priority queue that sorts non-seed nodes in descending order \\
				& of marginal gains (w.r.t. seed set) \\ \hline
			\end{tabular}
		\end{center}
	\end{small}
	\vspace{-4mm}
\end{table}

The influence maximization is an \NP-hard problem, under both MIA and IC models \cite{CWW10,KKT03}.

In spite of the aforementioned computational challenges of influence estimation and maximization,
the following properties of the influence function, $\sigma(S)$ assist us in developing a {\em Greedy} Algorithm (presented in Algorithm~\ref{alg:greedy}) with approximation guarantee of $(1-\frac{1}{e})$ \cite{KKT03}
\vspace{-1mm}
\begin{lma} \label{lma:submodularity}
Influence function is sub-modular \cite{KKT03,CWW10}.
A function $f$ is sub-modular if $f(S \cup \{x\}) - f(S) \geq f(T \cup \{x\}) - f(T)$ for any $x$, when $S\subseteq T$.
\end{lma}
\vspace{-3mm}
\begin{lma} \label{lma:monotone}
Influence function is monotone \cite{KKT03,CWW10}.
A function $f$ is monotone if $f(S\cup\{x\}) \geq f(S)$ for any x.
\end{lma}
The Greedy algorithm repeatedly selects
the node with the maximum marginal influence gain (line 2), and adds it to the current seed set (line 3)
until $k$ nodes are identified.

As given in Table~\ref{tab:notation}, we denote by $S^{i-1}$ the seed set formed at the end of the $(i-1)$-th iteration of Greedy, whereas $s_i \in S$
is the seed node added at the $i$-th iteration. Clearly, $1 \le i \le k$. One can verify that the following inequality holds for all $i$, $1\le i < k$.
\begin{align} \label{eq:greedy_marginal}
	MG(S^{i-1},s^i) \ge MG(S^i,s^{i+1})
	\vspace{-1mm}
\end{align}
\vspace{-3mm}
\subsection{IM in Dynamic Graphs}
\label{sec:dynIM}
Classical influence maximization techniques are developed for static graphs. The real-time influence graphs, however, are seldom static and evolves over time with various graph updates.

\spara{Graph update categories.} We recognize six update operations among which four are edge operations and two are node operations in dynamic graphs:
1. increase in edge probability, 2. adding a new edge, 3. adding a new node, 4. decrease in edge probability,
5. deleting an existing edge, and 6. deleting an existing node. We refer to the first three update operations as {\em additive updates},
because the size of the graph and its parameters increase with these operations; and the remaining as {\em reductive updates}.
Hereafter, we use a general term {\em update} for any of the above operations, until and unless specified, and we denote an update operation with $o$.

\spara{Dynamic influence maximization problem.}
\vspace{-1mm}
\begin{problem}
	\label{prob:dyn_inf}
	Given an initial uncertain graph $\mathcal{G}(V,E,P)$, old set $S^*_{old}$ of top-$k$ seed nodes, and a series of consecutive graph updates $\{o_1,o_2,\ldots,o_t\}$,
	find the new set $S^*_{new}$ of top-$k$ seed nodes for this updated graph.
\end{problem}
\vspace{-1mm}
The baseline method to solve the dynamic influence maximization problem will be to find the updated graph at every time, and then execute an IM
algorithm on the updated graph, which returns the new top-$k$ seed nodes. However, {\em computing all seed nodes from ground
	at every snapshot is prohibitively expensive}, even for moderate size graphs \cite{SLCHT17,OAYK16}.
Hence, our work aims at {\em incrementally updating the seed set}, without explicitly running the
complete IM algorithm at every snapshot of the evolving graph.

%% file: Prop_solution3.tex
\vspace{-2mm}
\section{Approximate Solution: MIA Model}
\label{sec:prop_sol}
We propose a novel {\sf N-Family} framework for dynamic influence maximization,
which can be adapted to many influence maximization algorithms and several influence diffusion models.
We first introduce our framework under the MIA model that illustrates how an update affects the nodes in the graph
(Section~\ref{sec:alg_affected}), and how to re-adjust the top-$k$ seed nodes with a theoretical performance
guarantee (Section~\ref{sec:seed_update}). Initially, we explain our technique for a single dynamic update,
and later we show how it can be extended to batch updates (Section~\ref{sec:batch_update}).
In Section~\ref{sec:IC_model}, we show how to extend our algorithm to IC and LT models, with efficient heuristics.
\vspace{-2mm}
\subsection{Finding Affected Regions}
\label{sec:alg_affected}
\vspace{-1mm}
Given an update, the influence spread of several nodes in the graph could be affected.
However, the nearby nodes would be impacted heavily, compared to a distant node.
We, therefore, design a threshold ($\theta$)-based approach to find the affected regions,
and our method is consistent with the notion of the MIA model. Recall that in MIA model,
an influence threshold $\theta$ is used to eliminate maximum influence paths that have
smaller propagation probabilities than $\theta$. Clearly, $\theta$ is an
input parameter to trade off between efficiency and accuracy, and its optimal value
is decided empirically.
\vspace{-1mm}
\begin{problem}
	\label{prob:affected}
	Given an update operation $o$ in an uncertain graph $\mathcal{G}(V,E,P)$, find all nodes $v \in V$ for which the expected influence spread $\sigma(\{v\})$ is changed by at least $\theta$.
\end{problem}
\vspace{-1mm}
In MIA model, the affected nodes could be computed exactly in polynomial time (e.g., by exactly finding
the expected influence spread of each node before and after the update, with the MIA model). In this work, we, however,
consider a more efficient upper bounding technique as discussed next.
\vspace{-1mm}
\subsubsection{Definitions}
We start with a few definitions.
\vspace{-1mm}
\begin{defn}[Maximum Influence In-Arborescence]
	Maximum Influence In-Arborescence (MIIA) \cite{CWW10} of a node $u\in V$ is the union of all the maximum influence paths to $u$ where every node in that path reaches $u$ with a minimum propagation probability of $\theta$, and it is denoted as $MIIA(u,\theta)$. Formally,
\vspace{-1mm}
\begin{align}
		MIIA(u,\theta) = \underset{v\in V}{\cup}\{MIP(v,u): \prod_{e\in MIP(v,u)} P_e \ge \theta\}
\end{align}
\end{defn}
\vspace{-2mm}
\begin{defn}[Maximum Influence Out-Arborescence]
Maximum Influence Out-Arborescence (MIOA) \cite{CWW10} of a node $u\in V$ is the union of all the maximum influence paths from $u$ where $u$ can reach every node in that path with a minimum propagation probability of $\theta$, and it is denoted as $MIOA(u,\theta)$.
\vspace{-1mm}
\begin{align}
		MIOA(u,\theta) = \underset{w\in V}{\cup}\{MIP(u,w): \prod_{e\in MIP(u,w)} P_e \ge \theta\}
\end{align}
\end{defn}
\vspace{-3mm}
\begin{defn}[{\sf 1-Family}]
	For every node $u \in V$, {{\sf 1-Family}} of $u$, denoted as $F_1(u)$, is the set of nodes that influence $u$, or get influenced by $u$ with
	minimum probability $\theta$ through the maximum influence paths, i.e.,
\vspace{-1mm}
\begin{align}
		F_1(u) = MIIA(u,\theta) \cup MIOA(u,\theta)
\end{align}
\end{defn}
\begin{defn}[{\sf 2-Family}]
For every node $u \in V$, {{\sf 2-Family}} of $u$, denoted as $F_2(u)$, is the union of the set of nodes present in {\sf 1-Family} of every node in $F_1(u)$, i.e.,
\vspace{-1mm}
\begin{align}
		F_2(u) = \underset{w\in F_1(u)}{\cup}{F_1(w)}
\end{align}
\end{defn}
Note that {\sf 2-Family} is always a superset of {\sf 1-Family} of a node.
\begin{exple}\label{ex:prtset}
	In Figure~\ref{fig:Inf_Prop}, let us consider $\theta = 0.07$. Then, $pp(\{C\},$ $C) = 1, pp(\{A\},C) = 0.8$, and $pp(\{L\},C) = 0.8\times0.1 = 0.08$. For any other node in the graph,
	its influence on $C$ is 0. Hence, $MIIA(C,0.07)=\{C, A, L \}$. Similarly, $MIOA(C,0.07)=\{C, D, E \}$.
	$F_1(C)$ will contain $\{C, A, L, D, E\}$. Analogously $F_2(C)$ will contain $\{C, A, L, D, E, B\}$. Since the context is clear, for brevity
	we omit $\theta$ from the notation of family.
\end{exple}
\vspace{-1mm}
We note that Dijkstra's shortest path algorithm, with time complexity $\bigO(|E| + |V|\log|V|)$ \cite{CWW10}, can be used to identify the $MIIA$, $MIOA$, and {\sf 1-Family} of a node. The time complexity for computing {\sf 2-Family} is $\bigO(|E| + |V|\log|V|)^2$. For simplicity, we refer to {\sf 1-Family} of a node as its family.

The {\sf 2-Family} of a seed node satisfies an interesting property (given in Lemma~\ref{lma:rem_mg}) in terms of marginal gains.
%
\begin{lma}\label{lma:rem_mg}
	Consider $s \in S$, then removing $s$ from the seed set $S$ does not change the marginal gain of any node that is not in $F_2(s)$. Formally, $MG(S,u)=MG(S\setminus\{s\},u)$,
	for all $u \in V \setminus F_2(s)$, according to the MIA model.
\end{lma}
\vspace{-1mm}

Formal proofs of all our lemma and theorems are given in the Appendix.
Intuitively, Lemma~\ref{lma:rem_mg} holds because the marginal gain of a node $u$ depends on the influence of seed nodes over those nodes that $u$ influences. For a node $u$ that is outside $F_2(s)$, there is no node that can be influenced by both $s$ and $u$.
It follows from the fact that a node influences, or gets influenced by the nodes that are present
only in its family, based on the MIA model.

\spara{Change in family after an update.}
During the additive update, e.g., an edge addition, the size of the family of a node nearby the update may increase. A new edge would help in more influence spread, as demonstrated below.
\vspace{-2mm}
\begin{exple}
	Consider Figure~\ref{fig:Inf_Prop_2} as the initial graph. When $\theta=0.07$, $F_1(A) = \{A, L, B\}$.
	Let us assume that a new edge $AC$ with probability $0.8$ is added, that is, the updated graph is
	now Figure~\ref{fig:Inf_Prop}. If we recompute $F_1(A)$ in Figure~\ref{fig:Inf_Prop}, then we get $F_1(A) = \{A, B, C, D, L\}$.
\end{exple}
\vspace{-1mm}
Analogously, during the reductive update, e.g., an edge deletion, the size of family of a node surrounding the update may decrease.
Deleting the edge eliminates paths for influence spread, as follows.
\vspace{-4mm}
\begin{exple}
	Consider Figure~\ref{fig:Inf_Prop} as the initial graph. $F_1(A) = \{A, B, C,$ $D, L\}$. Now, the edge $AC$ with probability $0.8$ is deleted.
	If we recompute $F_1(A)$ after modifying the graph (i.e., Figure~\ref{fig:Inf_Prop_2}), we get $F_1(A) = \{A, B, L\}$.
\end{exple}
\vspace{-1mm}
Thus, for soundness, in case of an additive update, we compute $MIIA$, $MIOA$, and family on the updated graph.
On the contrary, for a reductive update, we compute them on the old graph, i.e., before the update.
Next, we show in Lemma~\ref{lma:FIR3} that $MIIA(u,\theta)$ provides a safe bound on affected region for any
update originating at node $u$, according to the MIA model.
\vspace{-1mm}
\begin{lma}\label{lma:FIR3}
	In an influence graph $\mathcal{G}(V,E,P)$, adding a new edge $uv$ does not change the influence spread of any node outside $MIIA(u,$ $\theta)$ by
	more than $\theta$, according to the MIA model.
\end{lma}

\vspace{-1mm}
Lemma~\ref{lma:FIR3} holds because a node $u$ cannot be influenced by any node that is not in $MIIA(u, \theta)$, according to the MIA model. Hence, adding an edge
$uv$ does not change the influence spread (at all) of any node outside $MIIA(u, \theta)$.
This phenomenon can be extended to edge deletion, edge probability increase, and for edge probability decrease.
Moreover, for a node update (both addition and deletion) $u$, $MIIA(u,\theta)$ gives a safe upper bound of the affected region.
We omit the proof due to brevity. Therefore, $MIIA(u,\theta)$ is an efficient (computing time $\bigO(|E| + |V|\log|V|)$) and a
safe upper bound for the affected region.
\vspace{-2mm}
\subsubsection{Infected Regions}
\label{sec:inf_reg}
Due to an update in the graph, we find that a node may get affected in two ways:
(1) the nodes (including a few old seed nodes) whose influence spreads are significantly affected due
to the update operation, and also (2) those nodes whose marginal gains might change due to an affected
seed node, discovered in the previous step(s). This gives rise to a recursive definition, and multiple
levels of {\em infected} regions, as introduced next.

\spara{First infected region ({\sf 1-IR}).}
Whenever an update operation $o$ takes place, the influence spread of the nodes surrounding it, will change.
Hence, we consider the first infected region as the set of nodes, whose influence spreads change at least by $\theta$.
\vspace{-2mm}
\begin{defn}[First infected region ({\sf 1-IR})]
	In an influence graph $\mathcal{G}(V,E,P)$ and given a probability threshold $\theta$,
	for an update operation $o$, {\sf 1-IR}$(o)$ is the set of nodes whose influence spread changes greater than or equal to $\theta$. Formally,
	\vspace{-1mm}
	\begin{align}
		\text{\sf 1-IR}(o) = \{v \in V: |\sigma_{\mathcal{G}}(v)-\sigma_{\mathcal{G},o}(v)|\ge \theta\} \label{eq:1IR}
	\end{align}
\end{defn}
\vspace{-1mm}
In the above equation, $\sigma_\mathcal{G}(v)$ denotes the expected influence spread of $v$ in $\mathcal{G}$, whereas
$\sigma_{\mathcal{G},o}(v)$ is the expected influence spread of $v$ in the updated graph. Following our earlier discussion,
we consider $MIIA(u,\theta)$ as a proxy for {\sf 1-IR}$(o)$, where $u$ is the starting node for the update operation $o$.
\vspace{-2mm}
\begin{exple}
	In Figure~\ref{fig:Inf_Prop}, consider the removal of edge $AC$.
	Assuming $\theta=0.07$, {\sf 1-IR}$(o)$=$MIIA(A,0.07) =\{A, L \}$.
\end{exple}
\vspace{-2mm}
\spara{Second infected region ({\sf 2-IR}).}
We next demonstrate how infection propagates from the first infected region to other parts of the graph through the {\em family} of
affected seed nodes.

First, consider a seed node $s \in S$, a non-seed node $u \not \in S$,
and $s \in F_2(u)$. If the influence spread of $u$ has increased due to an update,
then to ensure that $s$ continues as a seed node, we have to remove $s$ from the seed set,
and recompute the marginal gain of every node in $F_2(s)$. The node, which has the maximum gain,
will be the new seed node. Second, if a seed node $s$ gets removed from the seed set in this
process, the marginal gains of all nodes present in $F_2(s)$ will change. We are now ready to define
the second infected region.
\vspace{-2mm}
\begin{defn}[Second infected region ({\sf 2-IR})]
	For an additive update $(o_a)$, the influence spread of every node present in {\sf 1-IR}$(o_a)$ increases which gives the possibility for any node in {\sf 1-IR} to become a seed node.
	Hence, the union of 2-Family of all the nodes present in {\sf 1-IR}$(o_a)$ is called the second infected region
	{\sf 2-IR}$(o_a)$. On the contrary, in a reductive update operation $o_r$, there is no increase in influence spread of any node in
	{\sf 1-IR}$(o_r)$. Hence, the union of 2-Family of old seed nodes present in {\sf 1-IR}$(o_r)$ is considered as the second infected
	region {\sf 2-IR}$(o_r)$.
	\begin{align}
		& \text{\sf 2-IR}(o_r) = \{ F_2(s) : s \in \text{\sf 1-IR}(o_r) \cap S\} & \label{eq:2IR_r}\\
		& \text{\sf 2-IR}(o_a) = \{F_2(u) : u\in \text{ \sf 1-IR}(o_a) \}\label{eq:2IR_a}
	\end{align}
\end{defn}
The time complexity to identify {\sf 2-IR} is $\bigO(m(|E| + |V|\log|V|)^2)$, where $m$ is the number of nodes in {\sf 1-IR}.
\vspace{-2mm}
\begin{exple}
	In Figure~\ref{fig:Inf_Prop}, consider the removal of edge $AC$.
	Assuming $\theta=0.07$, {\sf 2-IR}$(o)$=$F_2(A)$. This is because
	$A$ is an old seed node present in {\sf 1-IR}$(o)$ for this
	reductive update. Furthermore, because this is a reductive
	update, the family of $A$ needs to be computed before the update.
	Therefore, {\sf 2-IR}$(o)$=$F_2(A)=\{A,B,C,D,L,E\}$.
\end{exple}
\vspace{-2mm}
\spara{Iterative infection propagation.}
Whenever there is an update, the infection propagates through the 2-Family of the nodes whose marginal gain changes as discussed above.
For $N\ge 3$, the infection propagates from the $(N-1)^{th}$ infected region to the $N^{th}$ infected region through old seed nodes that are
present in the 2-Family of nodes in {\sf (N-1)-IR}.
\vspace{-2mm}
\begin{defn}[$N(\ge 3)$ infected region ({\sf N-IR})]
	The 2-Family of seed nodes, that are in the 2-Family of infected nodes in {\sf (N-1)-IR}, constitute
	the $N^{th}$ infected region.
	\vspace{-1.5mm}
	\begin{align}
		\text{\sf N-IR} = \{F_2(s): s \in F_2(u) \cap S, u\in {\text{\sf (N-1)-IR}}\}
		\label{eq:nIR}
	\end{align}
\end{defn}
\begin{figure}[t!]
	\centering
	\includegraphics[scale=0.21]{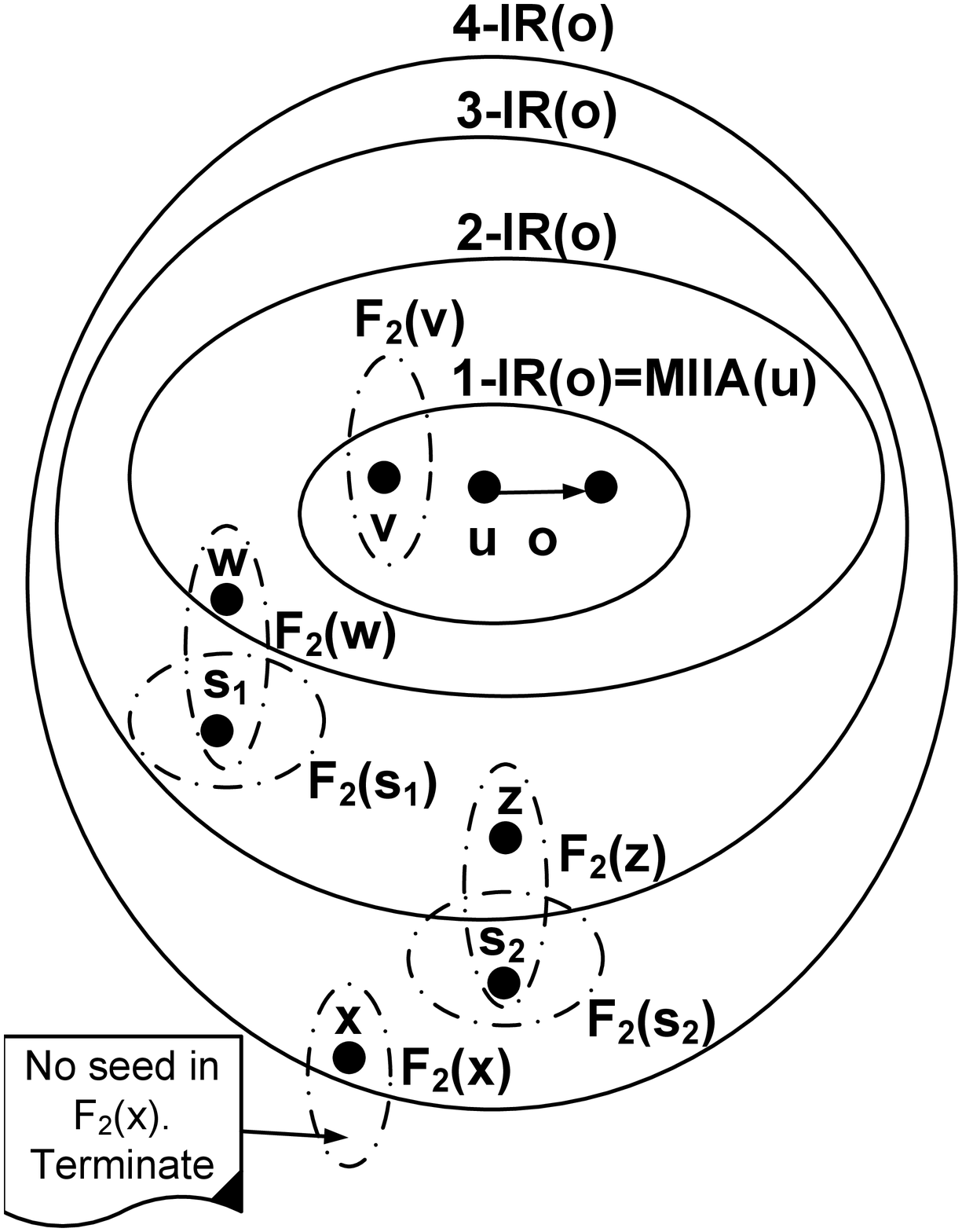}
	\vspace{-4mm}
	\caption{\small Iterative infection propagation: $o$ is an additive update operation originating at node $u$.
		$s_1$ and $s_2$ are two old seed nodes. $v$, $w$, $z$, $x$ are nodes, not necessarily old seed nodes.}
	\label{fig:tir}
	\vspace{-6mm}
\end{figure}
We demonstrate the iterative computation of infected regions, up to {\sf 4-IR} for an additive update, in Figure~\ref{fig:tir}. We begin with node $u$ which is the starting node of the update, and $MIIA(u,\theta)$ is the {\sf 1-IR}. The update being an additive one, union of 2-Family of all the nodes $v \in \text{\sf 1-IR}$ is considered as the {\sf 2-IR}. For all nodes $w \in \text{\sf 2-IR}$, we compute $F_2(w)$. Now, union of 2-Family of all seed nodes $s_1 \in F_2(w)$ is considered as {\sf 3-IR}. Similarly, {\sf 4-IR} can be deduced, and as there is no seed node present in the 2-Family of all nodes $x\in \text{\sf 4-IR}$, we terminate the infection propagation.

\spara{Termination of infection propagation.}
The infection propagation stops when no further old seed node is identified in the 2-Family of any node in the $N^{th}$ infected region. Due to this, there shall be no infected node present in 2-Family of any uninfected seed node.
For a seed set of cardinality $k$, it can be verified that the maximum value of $N$ can be between $1$ and $(k+1)$ for reductive update and between $2$ and $(k+2)$ for additive update.

\spara{Total infected region ({\sf TIR}).}
The union of all infected regions is referred to as the total infected region ({\sf TIR}).
\vspace{-1.5mm}
\begin{align}
{\sf TIR} = {\text{\sf 1-IR}} \cup {\text{\sf 2-IR}} \cup {\text{\sf 3-IR}} \cup \ldots \text{until termination} \label{eq:TIR}
\end{align}

Our recursive definition of {\sf TIR} ensures the following properties.
\vspace{-5mm}
\begin{lma} \label{lma:tir_out}
	The marginal gain of every node outside {\sf TIR} does not change, according to the MIA model.
	Formally, let $S$ be the old seed set, and
	we denote by $S_{rem}$ the \underline{rem}aining old seed nodes outside {\sf TIR},
	i.e., $S_{rem}=S\setminus {\text{\sf TIR}}$. Then, the following holds:
	$MG(S,v)=MG(S_{rem},v)$, for all nodes $v \in V\setminus{\text{\sf TIR}}$.
\end{lma}
\vspace{-1mm}
Lemma~\ref{lma:tir_out} holds because any node outside {\sf TIR} does not belong to {\sf 2-Family} of any seed node present in TIR. Hence, by Lemma~\ref{lma:rem_mg}, its marginal gain does not change.
\vspace{-2mm}
\begin{lma} \label{lma:tir_in_out}
	Any old seed node outside {\sf TIR} has no influence on the nodes inside {\sf TIR}, following the MIA model.
	Formally, $pp(S_{rem},u)=0$, for all nodes $u \in {\text{\sf TIR}}$.
\end{lma}
\vspace{-1mm}
Lemma~\ref{lma:tir_in_out} holds because any uninfected seed node is more than 2-Family away from any node present in {\sf TIR} (This is how we terminate infection propagation). Hence, there is no node present in {\sf TIR} that belongs to the family of any seed node outside {\sf TIR}.
\begin{algorithm}[tb!]
	\caption{\,\,\,{\sf N-Family} seeds updating method on top of {\sf Greedy}}
    \label{alg:prop}
	\begin{algorithmic}[1]
		\Require {Graph $\mathcal{G}(V, E, P)$, total infected region {\sf TIR}, old seed set $S$, $|S| = k$, old priority queue $Q$}
		\Ensure {Compute the new seed set $S_{new}$} of size $k$
		\State $S_{rem} \leftarrow S \setminus {\sf TIR}$
		\ForAll {$u \in {\sf TIR}$}
		\State  $Q(u) \leftarrow \sigma(u)$
		\EndFor
		\While {TRUE}
		\State $S_{new} =$ {Greedy}$(\mathcal{G},S_{rem}, k)$ \\
		/\/* Starting w/\ seed set $S_{rem}$, add $k-|S_{rem}|$ seeds by {Greedy} */\
		\State $S_{order} \leftarrow$ Sort nodes in $S_{new}$ in {Greedy} inclusion order
		\State $w \leftarrow Q$[top]
		\If {$(S^{k-1}_{order},s^k_o) < MG(S_{order},w)$}
		\ForAll {$u \in F_2(s^k_o)\setminus S^{k-1}_{order}$}
		\State $Q(u) \leftarrow MG(S^{k-1}_{order},u)$
		\State $S_{rem} \leftarrow S^{k-1}_{order}$
		\EndFor
		\Else
		\State Output {$S_{order}$}
		\EndIf
		\EndWhile
	\end{algorithmic}
\end{algorithm}

The old seed nodes inside {\sf TIR}
may no longer continue as seeds, therefore we need to discard them from the seed set, and the same number of
new seed nodes have to be identified.
We discuss the updating procedure of seed nodes in the following section.
\vspace{-3mm}
\subsection{Updating the Seed Nodes}
\label{sec:seed_update}
\vspace{-1mm}
We now describe our seed updating method over the Greedy IM algorithm,
under the MIA model of influence cascade. Later we prove that the new seed nodes reported by our technique (Algorithm~\ref{alg:prop}) will be the same as the top-$k$ seed nodes found by Greedy
on the updated graph and with the MIA model, thereby maintaining $(1-\frac{1}{e})$ approximation guarantee to the optimal solution \cite{CWW10}.
\vspace{-4mm}
\subsubsection{Approximation Algorithm}
\label{sec:update_algo}
\vspace{-1mm}
We present our proposed algorithm for updating the seed set in Algorithm~\ref{alg:prop}. Consider Greedy (Algorithm~\ref{alg:greedy}) over the MIA model on the initial graph, and assume that we obtained the seed set $S$, having cardinality $k$.
Since Greedy works in an iterative manner, let us denote by $S^{i-1}$ the seed set formed at the end of the $(i-1)$-th iteration, whereas $s_i \in S$
is the seed node added at the $i$-th iteration. Clearly, $1\le i \le k$, $|S^{i-1}|=i-1$, and $S=S^k=\cup_{i=1}^k s_i$.
Additionally, we use a priority queue $Q$, where its top node $w$ has the maximum marginal gain $MG(S,w)$ among all the non-seed nodes.

After the update $o$, we first compute the total infected region, {\sf TIR} using Equation~\ref{eq:TIR}.
Consider $S_{rem}$, of size $|S_{rem}|=k'$, as the set of old seed nodes outside {\sf TIR}, i.e.,
$S_{rem}=S\setminus {\text{\sf TIR}}$. Then, we remove $(k-k')$ old seed nodes inside {\sf TIR},
and our next objective is to identify $(k - k')$ new seed nodes from the updated graph.

Note that inside $S_{rem}$, the seed nodes are still sorted in descending order of their marginal gains,
computed at the time of insertion in the old seed set $S$ following the Greedy algorithm.
In particular, we denote by $s_{r}^j$ the $j$-th seed node in descending order inside $S_{rem}$, where $1\le j \le k'$.
Due to Lemma~\ref{lma:tir_out}, $MG(S,v)=MG(S_{rem},v)$, for all nodes $v \in V\setminus {\text{\sf TIR}}$.
Thus, for all $j$,  $1\le j <k'$ the following inequalities hold.
\begin{align}
	MG(S_{rem}^{j-1},s_r^j)\ge MG(S_{rem}^j, s_r^{j+1})\label{eq:order_rem} \\
	MG(S_{rem}^{k'-1},s_r^{k'}) \ge MG(S_{rem}^{k'-1},v) \label{eq:order_rem_2}
\end{align}

Now, after removing the old seed nodes present in {\sf TIR} from the seed set, we compute the influence spread $\sigma(u)$
of every node $u \in {\text{\sf TIR}}$ and, we update these nodes $u$ in the priority queue $Q$, based on their new marginal
gains $\sigma(u)$ (lines 1-4).
It can be verified that $MG(S_{rem},u)=\sigma(u)$, for all $u \in {\text{\sf TIR}}$, due to Lemma~\ref{lma:tir_in_out}.

Now, we proceed with greedy algorithm and find the new $(k-k')$ seed nodes. Let us denote by $S_{new}$ the new seed set (of size $k$) found in this manner (line 6). Next, we sort the seed nodes in $S_{new}$ in their {\em appropriate} inclusion order according to the Greedy algorithm over the updated graph
(line 7).
This can be efficiently achieved by running Greedy {\em only} over the seed nodes in $S_{new}$, while computing their influence spreads
and marginal gains in the updated graph. The sorted seed set is denoted by $S_{order}$.
Let us denote by $s_o^k$ the last (i.e., $k$-th) seed node in $S_{order}$, whereas $S_{order}^{k-1}$ represents
the set of top-$(k-1)$ seed nodes in $S_{order}$. We denote by $w$ the top-most seed node in the priority queue
$Q$. If $MG(S_{order}^{k-1},s_o^k) \ge MG(S_{order},w)$, we terminate our updating algorithm (line 15).

\vspace{-0.7mm}
\spara{Iterative seed replacement.}
On the other hand, if $MG(S_{order}^{k-1},s_o^k) \hspace{-1mm} < \hspace{-1mm} MG(S_{order},w)$, we remove the last seed node $s_o^k$
from $S_{order}$. For every node $u$ in the
$F_2(s_o^k)\setminus S_{order}^{k-1}$, we compute marginal gain
$MG(S_{order}^{k-1},u)$ and update the priority queue $Q$ (lines 10-11). Next, we compute a new seed node
using Greedy and add it to $S_{order}^{k-1}$, thereby updating the seed set $S_{order}$.
We also keep the nodes in $S_{order}$ sorted after every update in it. Now, we again verify the condition:
if $MG(S_{order}^{k-1},s_o^k) \hspace{-1mm} < \hspace{-1mm} MG(S_{order},w)$, where $w$ being the new top-most node in the priority queue $Q$, then we repeat the above steps,
each time replacing the last seed node $s_o^k$ from $S_{order}$, with the top-most node from the updated priority queue $Q$.
This iterative seed replacement phase terminates when $MG(S_{order}^{k-1},s_o^k) \hspace{-1mm} \ge \hspace{-1mm} MG(S_{order},w)$. Clearly, this
seed replacement can run for at most $|S_{rem}|=k'$ rounds; because in the worst
scenario, all old seed nodes in $S_{rem}$
could get replaced by new seed nodes from {\sf TIR}.
Finally, we report $S_{order}$ as the new seed set.
\vspace{-2mm}
\subsubsection{Theoretical Performance Guarantee}
\label{sec:perf:guar}
We show in the Appendix that the top-$k$ seed nodes reported by our {\sf N-Family} method are the same as the top-$k$ seed nodes obtained by running the Greedy
on the updated graph under the MIA model. Since, the Greedy algorithm provides the approximation guarantee of $1-\frac{1}{e}$
under the MIA model \cite{CWW10}, our {\sf N-Family} also provides the same approximation guarantee.
\vspace{-3mm}
\subsection{Extending to Batch Updates}
\label{sec:batch_update}
\vspace{-1mm}
We consider the difference of nodes and edges present in two snapshots at different time intervals of the
evolving network as a set of batch updates. Clearly, we consider only the
final updates present in the second snapshot, avoiding the intermediate ones.
For example, in between two snapshot graphs, if an edge $uv$ is added and then gets deleted, we will not consider it as an update because there is no change in the graph with respect to $uv$ after the final update.

One straightforward approach would be to apply our algorithm for every update
sequentially. However, we develop a more efficient technique as follows.
For a batch update consisting of $m$ individual updates,
every update $o_i$ has its own {\sf TIR}$(o_i), i = 1, 2, 3, \ldots, m$.
The {\sf TIR} of the batch update is the union of {\sf TIR}$(o_i)$, for all $i\in(1,m)$.
\begin{align}
	\displaystyle \text{\sf TIR} = \cup_{i=1}^m \text{\sf TIR}(o_i)
\end{align}

Once the {\sf TIR} is computed corresponding to a batch update, we update the seed set using Algorithm~\ref{alg:prop}.
Processing all the updates in one batch is more efficient than the sequential updates. For example, if a seed node is
affected multiple times during sequential updates, we have to check if it remains the seed node every time.
Whereas in batch update, we need to verify it only once.
\vspace{-2mm}
\section{Heuristic Solution: IC and LT models}
\label{sec:IC_model}
Here, we will show how one can develop efficient heuristics by extending the proposed {\sf N-Family} approach to IC and LT
models \cite{KKT03}. We start with IC model.

\spara{Computing TIR.}
For IC model, one generally does not use any probability threshold $\theta$ to discard
smaller influences; and perhaps more importantly, finding the nodes whose influence spread changes
by at least $\theta$ (due to an update operation) is a \sharpP-hard problem. Hence, computing
{\sf TIR} under IC model is hard as well, and one can {\em no longer ensure a theoretical performance
	guarantee} of $(1-\frac{1}{e})$ as earlier. Instead, we {\em estimate} {\sf TIR} analogous to MIA
model (discussed in Section~\ref{sec:inf_reg}), which generates high-quality results
as verified in our detailed empirical evaluation. This is because the maximum influence paths considered by MIA model
play a crucial role in influence cascade over real-world networks \cite{CWW10}.

\spara{Updating Seed set.}
Our method for updating the seed set in IC model follows the same outline as given in Algorithm~\ref{alg:prop} with two major differences. In lines 3 and 11 of Algorithm~\ref{alg:prop}, we compute the marginal gains and update the priority queue, but now we employ more efficient techniques based on the IM algorithm used for the purpose. In particular, as discussed next, we derive two efficient heuristics, namely, Family-CELF (or, F-CELF) and Family-RRS (or F-RRS) by employing our {\sf N-Family} approach on top of two efficient IM algorithms CELF \cite{LKGFVG07} and RR sketch \cite{BBCL14}, respectively.

\vspace{-2mm}
\subsection{N-Family for IM Algorithms in IC model}
First, we explain static IM algorithms briefly, and then we introduce the methods to adapt them to a dynamic setting.
\vspace{-2mm}
\subsubsection{CELF}
In the Greedy algorithm discussed in Section~\ref{sec:preliminaries},
marginal influence gains of all remaining nodes need to be repeatedly calculated at every round, which makes it inefficient
(see Line 3, Algorithm~\ref{alg:greedy}). However, due to the sub-modularity property of the influence function, the marginal gain of a node in the present iteration
cannot be more than that of the previous iteration. Therefore, the CELF algorithm \cite{LKGFVG07} maintains a priority queue containing the nodes and their
marginal gains in descending order. It associates a {\em flag} variable with every node, which stores the iteration number in which the
marginal gain for that node was last computed. In the beginning, (individual) influence spreads of all nodes are calculated and added to the priority queue,
and flag values of all nodes are initiated to zero. In the first iteration, the top node in the priority queue is removed, since it has the maximum influence spread,
and is added to the seed set. In each subsequent iteration, the algorithm takes the first element from the priority queue, and verifies the status of its flag.
If the marginal gain of the node was calculated in the current iteration, then it is considered as the next seed node;
else, it computes the marginal gain of the node, updates its flag, and re-inserts the node in the priority queue. This process repeats until $k$ seed nodes are identified.

\vspace{-1.2mm}
\spara{FAMILY-CELF}
We refer to the {\sf N-Family} algorithm over CELF as FAMILY-CELF (or, {\sf F-CELF}).
In particular, we employ MC-sampling to compute marginal gains in lines 3 and
11 of Algorithm~\ref{alg:prop}, and then update the priority queue.
Given a node $u$ and the current seed set $S$, the corresponding marginal gain can be derived
with two influence spread computations, i.e., $\sigma(S\cup\{u\})-\sigma(S)$. However,
thanks to the {\em lazy forward} optimization technique in CELF, one may insert
any upper bound of the marginal gain in the priority queue. The actual marginal
gain needs to be computed only when that node is in the top of the priority queue
at a later time. Therefore, we only compute the influence spread of $u$, i.e., $\sigma(\{u\})$,
which is an upper bound to its marginal gain, and insert this upper bound in the priority queue.
\vspace{-3mm}
\subsubsection{Reverse Reachable (RR) Sketch}
In this method, first proposed by Borgs et al. \cite{BBCL14} and later improved by Tang et al. \cite{TSX15,TXS14},
subgraphs are repeatedly constructed and stored as {\em sketches} in index $\mathbb{I}$. For each subgraph $H_i$, an arbitrary node $z_i \in V$, selected uniformly at random,
is considered as the target node. Using a reverse Breadth First Search (BFS) traversal, it finds all nodes that influence $z_i$ through {\em active edges}.
An activation function $x_i:E \rightarrow (0,1)$ is selected uniformly at random, and for each edge $uv \in E$, if $x_i(uv) \le P_{uv}$,
then it is considered active. The subgraph $H_i$ consists of all nodes that can influence $z_i$ via these active edges.
Each sketch is a tuple containing $(z_i, x_i, H_i)$. This process halts when the total number of edges examined exceeds a pre-defined
threshold $\tau = \Theta( \frac{1}{\epsilon^2}k(|V| + |E|)\log{|V|})$, where $\epsilon$ is an error function associated with
the desired quality guarantee $(1-\frac{1}{e}-\epsilon)$. The intuition is that if a node appears in a large number of subgraphs,
then it should have a high probability to activate many nodes, and therefore, it would be a good candidate for a seed node.
Once the sufficient number of sketches are created as above, a greedy algorithm repeatedly identifies the node present in the majority of sketches,
adds it to the seed set, and the sketches containing it are removed. This process continues until $k$ seed nodes are found.

\spara{FAMILY-RRS}
We denote the {\sf N-FAMILY} algorithm over RR-Sketch as FAMILY-RRS (or, {\sf F-RRS}).
RRS technique greedily
identifies the node present in the majority of sketches,
adds it to the seed set, and the sketches containing it are deleted.
This process continues until $k$ seed nodes are identified.
In our {\sf F-RRS} algorithm, instead of deleting sketches as above,
we remove them from $\mathbb{I}$, and store them in another index $\mathbb{R}$,
since these removed sketches could be used later in our seeds updating procedure.

Let $\mathbb{I}_v \subseteq \mathbb{I}$ be the set of sketches $(z,x,H) \in \mathbb{I}$
with $v \in H$. Similarly, $\mathbb{R}_u \subseteq \mathbb{R}$ represents the set of
all sketches $(z, x, H) \in \mathbb{R}$ with $u \in H$. Furthermore, $\mathbb{I}^S$ (similarly $\mathbb{R}^S$)
denotes $\mathbb{I}$ (similarly $\mathbb{R}$) after the seed set $S$ is identified.
Clearly, the sketches in $\mathbb{I}^S$ will not have any seed node in their subgraphs.
Also note that $MG(S,v)$ is proportional to $|\mathbb{I}^S_v|$, by following the RRS technique.

After an update operation, we need to modify the sketches (both in $\mathbb{I}$ and $\mathbb{R}$),
and also to possibly swap some sketches between these two indexes, as discussed next.

\spara{Modifying sketches after dynamic updates.}
In the following, we only discuss sketch updating techniques corresponding to an edge addition.
Sketch updating methods due other updates (e.g., node addition, edge deletion, etc.) are similar \cite{OAYK16},
and we omit them due to brevity. To this end, we present three operations:

\underline{Expanding sketches:}
Assume that we added a new edge $uv$. We examine every sketch $(z, x, H)$ both in $\mathbb{I}^S_v$ and $\mathbb{R}^S_v$,
and add every new node $w$ that can reach $v$ through active edges in $H$.
We compute these new nodes using a reverse breadth first search from $v$.
In this process, the initial subgraph $H$ is extended to $H^e$.

Next, we need to update $\mathbb{I}^S$ and $\mathbb{R}^S$ in such a way that
sketches in $\mathbb{I}^S$ do not have a seed node in their (extended) subgraphs.
For every sketch $(z, x, H^e) \in \mathbb{I}^S$, if $H^e \cap S \ne \phi$,
we then remove $(z, x, H^e)$ from $\mathbb{I}^S$, and add it to $\mathbb{R}^S$.

\underline{Deleting sketches:}
If the combined weight of indexes except the last sketch exceeds the
threshold ($\tau = \Theta( \frac{1}{\epsilon^2}k(|V| + |E|)\log{|V|})$), we delete the last sketch $(z, x, H)$ from the index
where it belongs to (i.e., either from $\mathbb{I}^S$ or $\mathbb{R}^S$).

\underline{Adding sketches:}
If the combined weight of indexes is less than the threshold $\tau$, we select a target node $z \in V$ uniformly at random,
and construct a new sketch ${(z, x, H)}$. If $H \cap S = \phi$, we add the new sketch to $\mathbb{I}^S$, otherwise to $\mathbb{R}^S$.

\spara{Sketch swapping for computing marginal gains.}
Assume that we computed {\sf TIR},  $S_{inf} = S\cap \text{\sf TIR}$, and $S_{rem}=S\setminus\text{\sf TIR}$.
For every infected old seed node $s\in S_{inf}$, we identify all sketches $(z, x, H)$ with $s \in H$, that are
present in $\mathbb{R}^S$. Then, we perform the following sketch swapping to ensure that all infected seed nodes are removed from the old seed set.
{\bf (1)} If there is no uninfected seed node in $H$ (i.e, $H \cap S_{rem} = \phi$), where $(z, x, H) \in \mathbb{R}^S$,
we move $(z, x, H)$ from $\mathbb{R}^{S}$ to $\mathbb{I}^{S}$. {\bf (2)} If there is an uninfected seed node in $H$ $(i.e., H \cap S_{rem} \ne \phi)$,
where $(z, x, H) \in \mathbb{R}^S$, we keep $(z, x, H)$ in $\mathbb{R}^{S}$.

Finally, we identify $(k-k')$ new seed nodes using updated $\mathbb{I}^{S}$.
Marginal gain computation at line 11 (Algorithm~\ref{alg:prop}) follows a similar
sketch replacement method, and we omit the details for brevity.
\begin{table} [tb!]
	\vspace{-3mm}
	\caption{\small Properties of datasets}
	\scriptsize
	\vspace{-4mm}
	\centering
	\begin{tabular} { |l|c|c|c|c| }
		{\textsf{Dataset}} & {\textsf{\#Nodes}}  & {\textsf{\#Edges}}  &  \multicolumn{2}{c|}{\textsf{ Timestamps}} \\
		\cline{4-5}
		&&& \textsf{From} &  \textsf{To}\\ \hline
		{\sf Digg}      & {30\,398}   & {85\,247}   & 10-05-2002         & 11-23-2015          \\
		{\sf Slashdot}	& {51,083}    & {130\,370}  & 11-30-2005         & 08-15-2006           \\
		{\sf Epinions}	& {131\,828}  & {840\,799}  & 01-09-2001         & 08-11-2003           \\
		{\sf Flickr}	& {2\,302\,925} & {33\,140\,017}& 11-01-2006     &  05-07-2007     \\ 
	\end{tabular}
	\vspace{-3mm}
	\label{tab:data}
	\vspace{-3mm}
\end{table}
\vspace{-6mm}
\subsection{Implementation with the LT Model}
\label{sec:apply_lt}
\vspace{-1mm}
As we discussed earlier, the {\sf N-Family} algorithm can be implemented on top of both Greedy and CELF. However, these IM algorithms
also work with the linear threshold (LT) model. Hence, our algorithm can be used with the LT model. We omit details due to brevity.
\vspace{-6mm}
\subsection{Heuristic TIR Finding to Improve Efficiency}
\label{sec:tech_imp_eff}
\vspace{-1mm}
We propose a more efficient heuristic method, by carefully tuning the parameters
(e.g., by limiting $N=1, 2$ in {\sf TIR} computation) of our {\sf N-Family} algorithm.
Based on our experimental analysis with several evolving networks, we find that the influence spread changes significantly only for those nodes
which are close to the update operation. Another seed node, which is far away from the update operation, even though its influence spread (and its marginal
gain) may change slightly, it almost always remains as a seed node in the updated graph.
Hence, we further improve the efficiency of our {\sf N-Family} algorithm  by limiting $N=1, 2$ in {\sf TIR} computation.
Indeed, the major difference in influence spreads between the new seed set and the old one comes from those seed nodes in the first two infected regions (i.e., {\sf 1-IR}
and {\sf 2-IR}), which can also be verified from our experimental results (Section~\ref{sec:sens_anl}).

%% file: experiments.tex
\vspace{-2mm}
\section{Experimental Results}
\label{sec:experiments}
\vspace{-1.5mm}
\subsection{Experimental Setup}
\label{exp setup}

\vspace{-1mm}
$\bullet$ {\bf Datasets.}
We download four real-world graphs (Table~\ref{tab:data}) from the Koblenz Network Collection (http://konect. uni-koblenz.de/ networks/). All
these graphs have directed edges, together with time-stamps; and hence, we consider them as evolving networks. If some edge appears for multiple times,
we only consider the first appearance of that edge as its insertion time in the graph. The edge counts in Table~\ref{tab:data} are given considering
distinct edges only.
\begin{figure*}[t!]
	\vspace{-2mm}
	\centering
	\subfigure[{\em Edge add., Digg (\textsf{DWA})}] {
		\includegraphics[scale=0.15, angle=270]{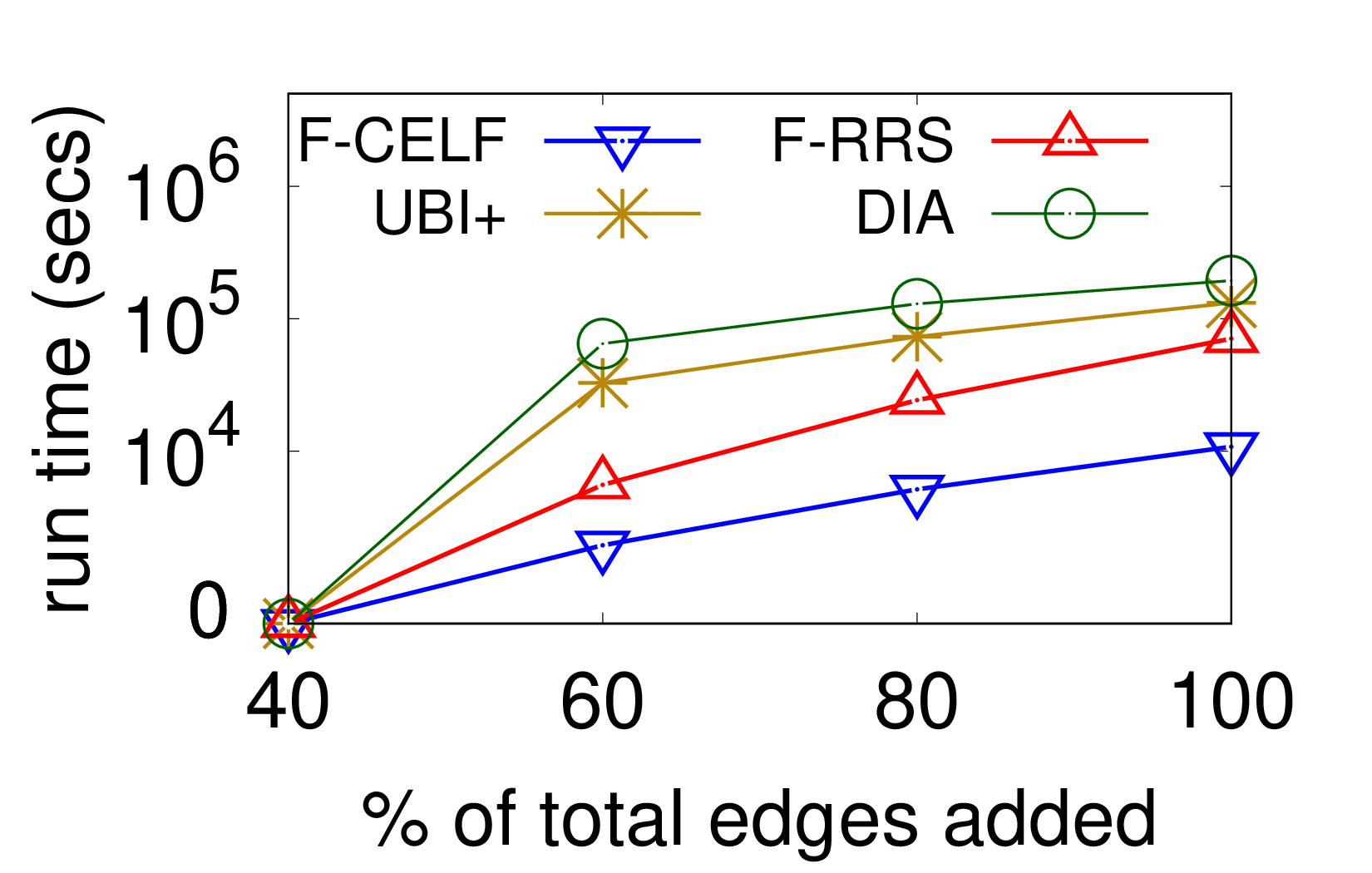}
		\label{fig:edgeadd_digg}
	}
	\subfigure[{\em Edge del., Slashdot (\textsf{TV})}]  {
		\includegraphics[scale=0.15, angle=270]{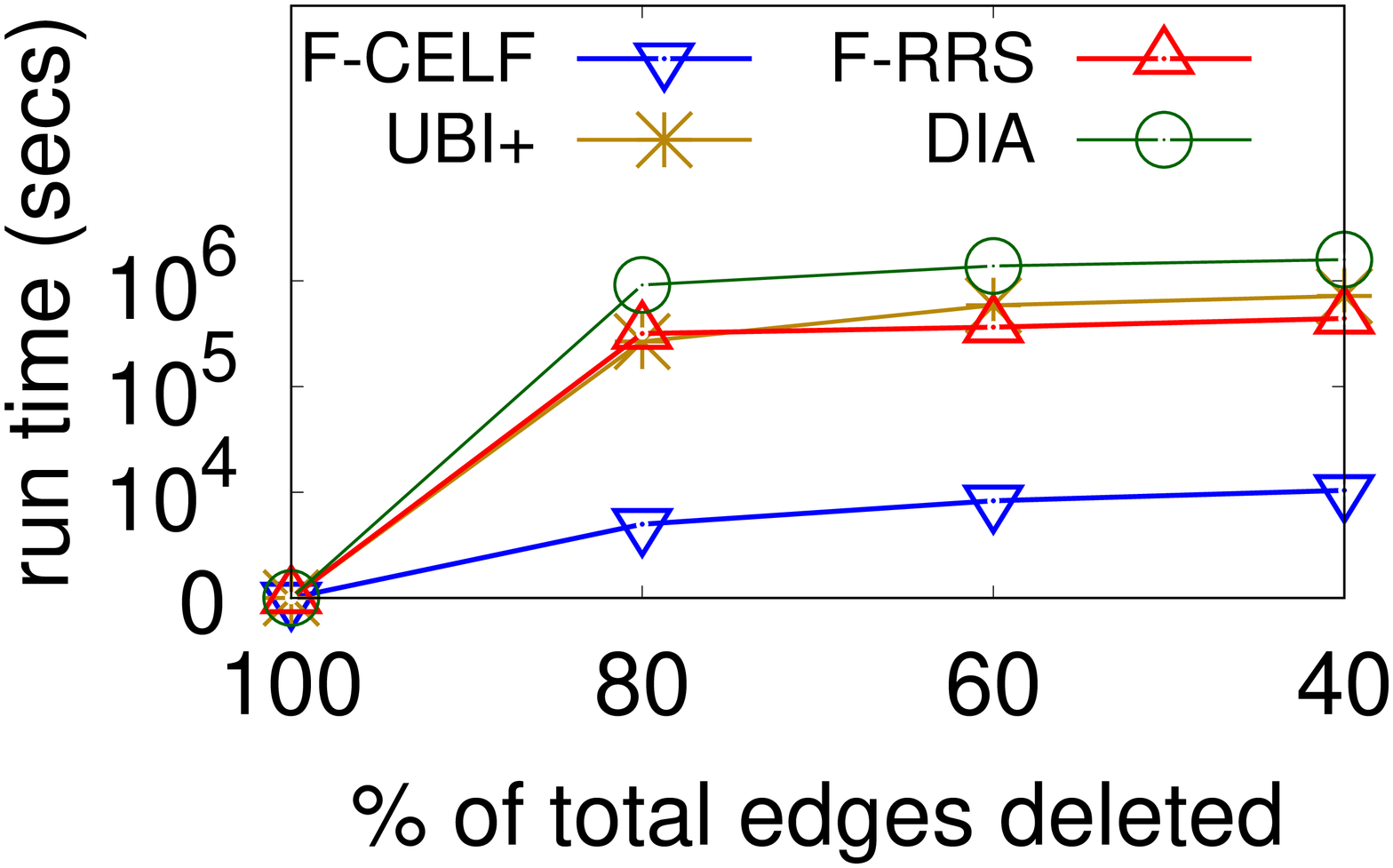}
		\label{fig:edgedel_slash}
	}
	\subfigure[{\em Node add., Epinions (\textsf{TV})}]  {
		\includegraphics[scale=0.15, angle=270]{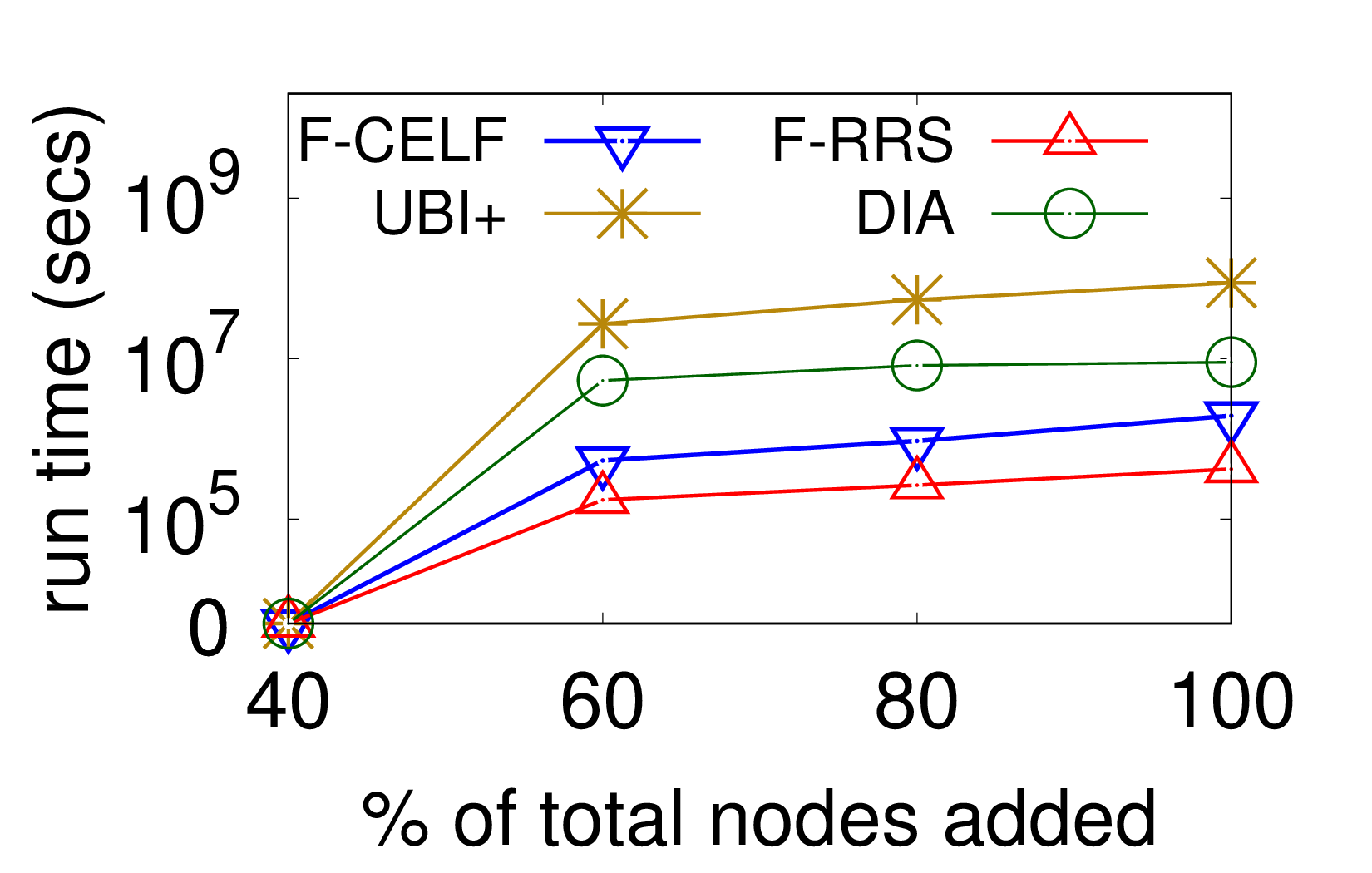}
		\label{fig:nodeadd_epin}
	}	
	\subfigure[{\em Node del., Flickr (\textsf{DWA})}] {
		\includegraphics[scale=0.15, angle=270]{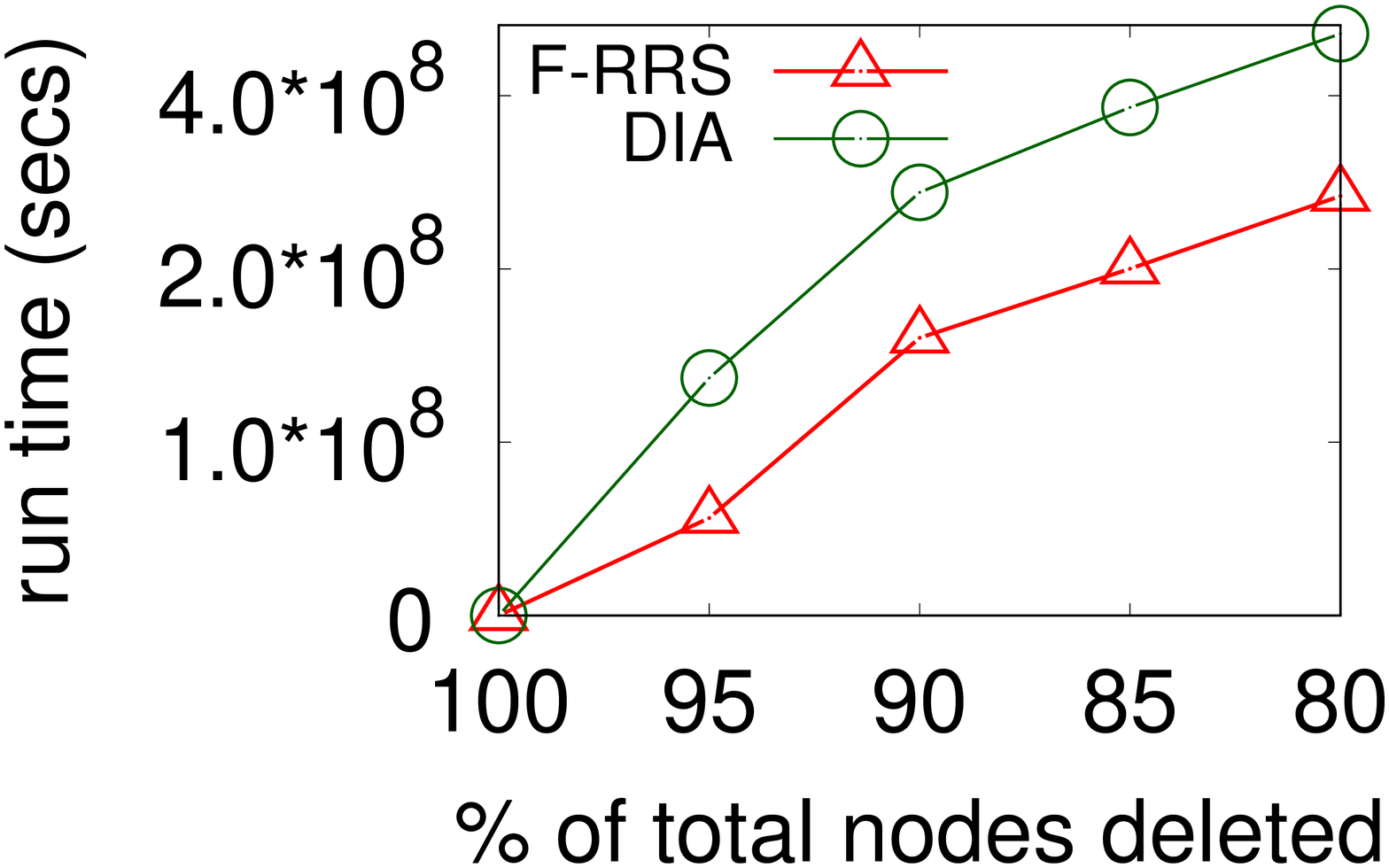}
		\label{fig:nodedel_flickr}
	}
	\vspace{-6mm}
	\caption{\small Run time to adjust seed set, IC model, seed sets are adjusted after every update}
	\label{fig:eff_ic}
	\vspace{-6mm}
\end{figure*}
\begin{figure}[t!]
	\vspace{-2mm}
	\centering
	\subfigure[{\em Node add., Slashdot (\textsf{DWA})}] {
		\includegraphics[scale=0.14, angle=270]{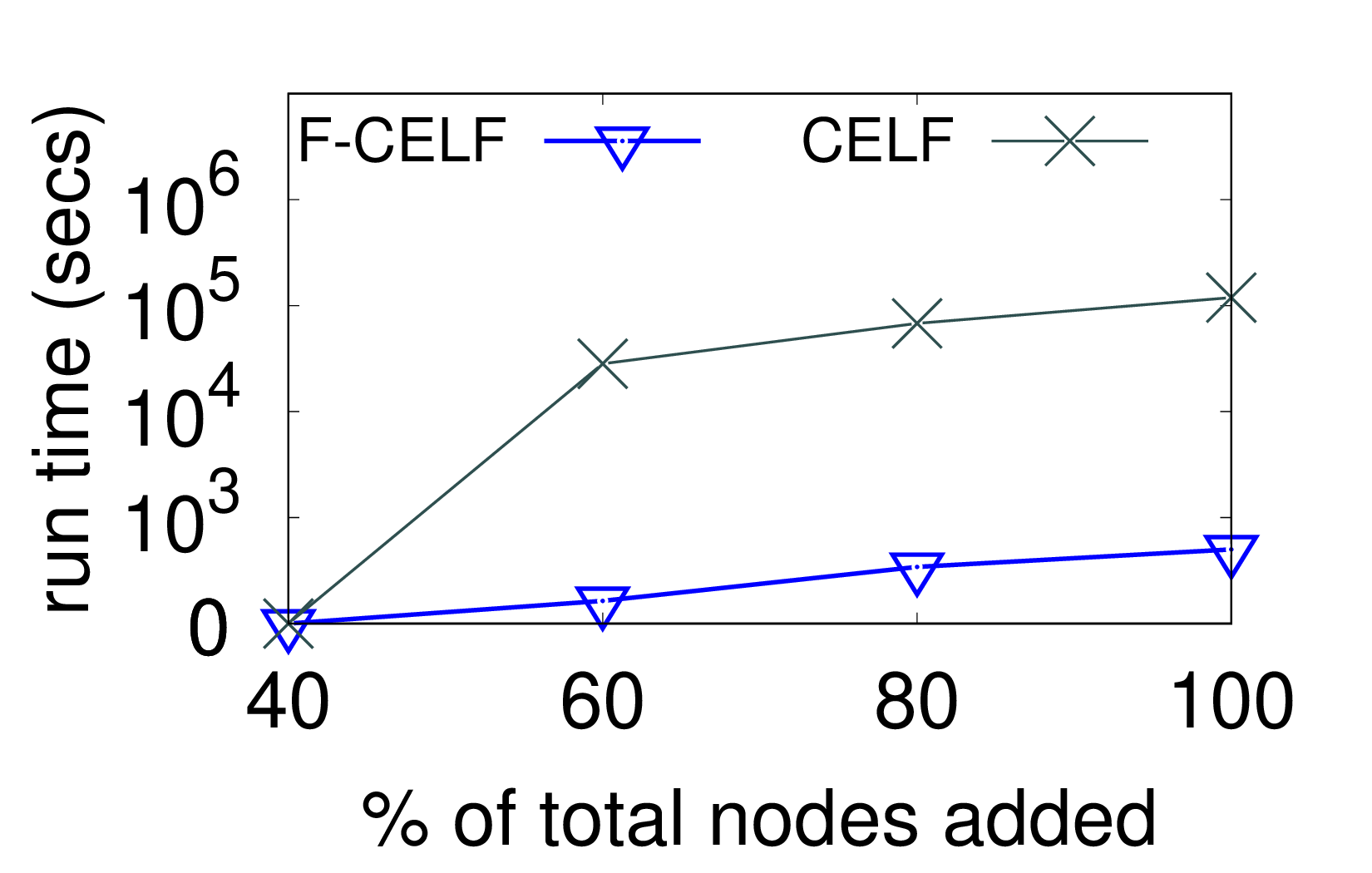}
		\label{fig:mia_nodeadd_slash}
	}
	\hspace{-1mm}
	\subfigure[{\em Node del., Epinions (\textsf{DWA})}]  {
		\includegraphics[scale=0.14, angle=270]{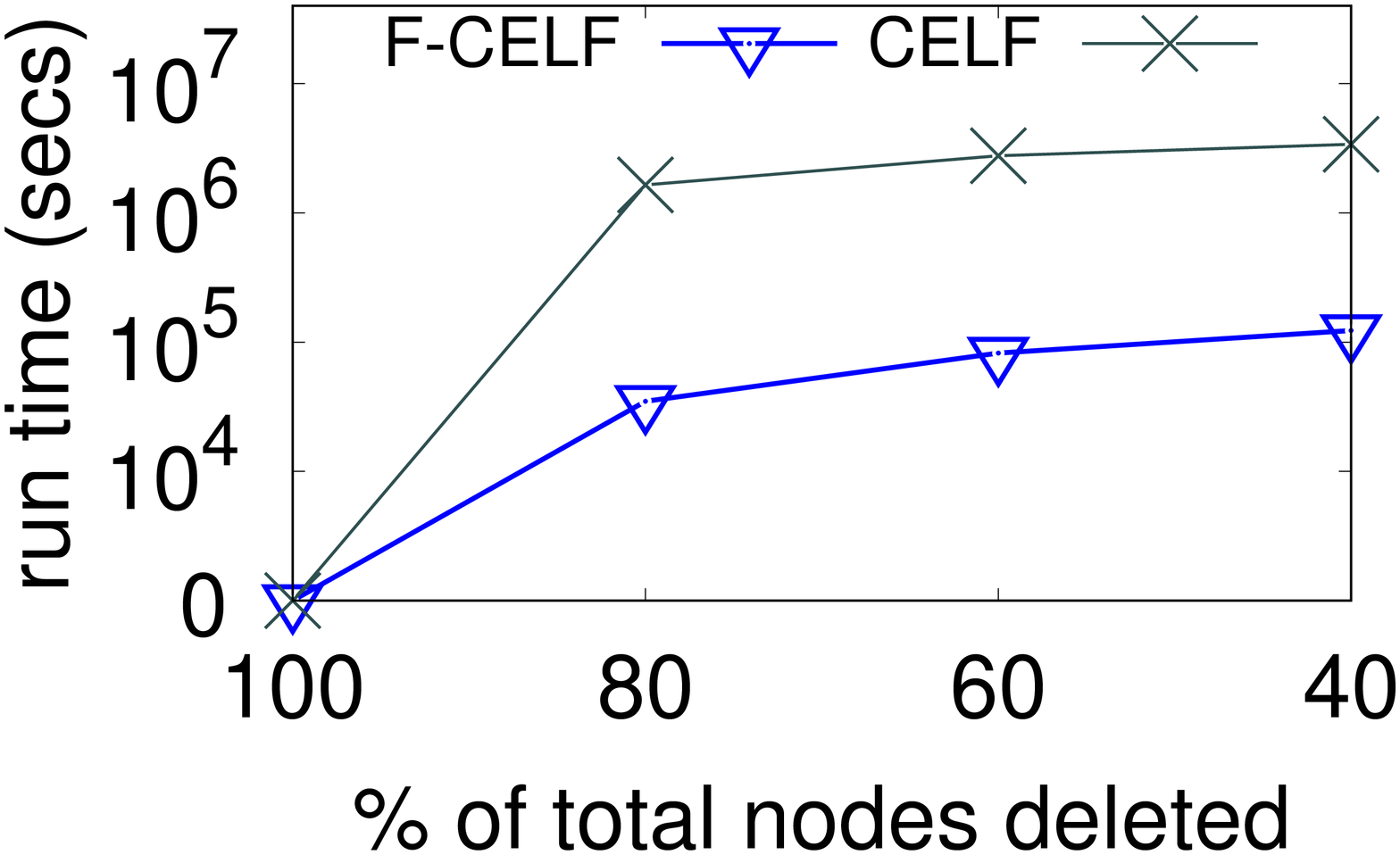}
		\label{fig:mia_nodedel_epin}
	}	
	\vspace{-6mm}
	\caption{\small Run time to adjust seed set, MIA model, seed sets are adjusted after every update}
	\label{fig:eff_MIA}
	\vspace{-6mm}
\end{figure}
\begin{figure}[t!]
	\vspace{-2mm}
	\centering
	\subfigure[{\em Inf. spread, edge add., \newline Digg (\textsf{DWA}) in IC model}]  {
		\includegraphics[scale=0.14, angle=270]{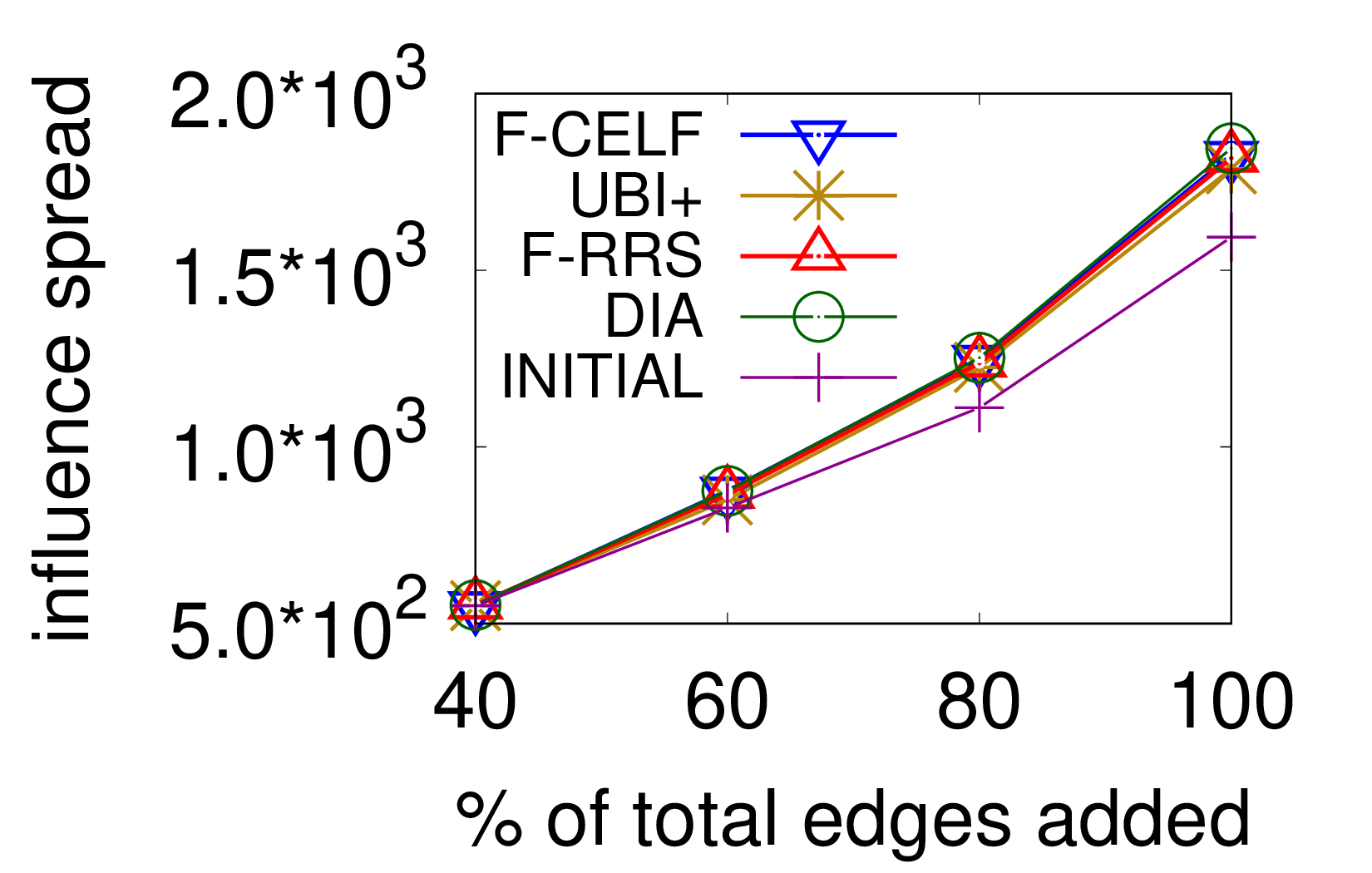}
		\label{fig:inf_sprd_digg}
	}\hspace{0mm}
	\subfigure[{\em Inf. spread, node del., \newline Epinions (\textsf{DWA}) in MIA model}] {
		\includegraphics[scale=0.14, angle=270]{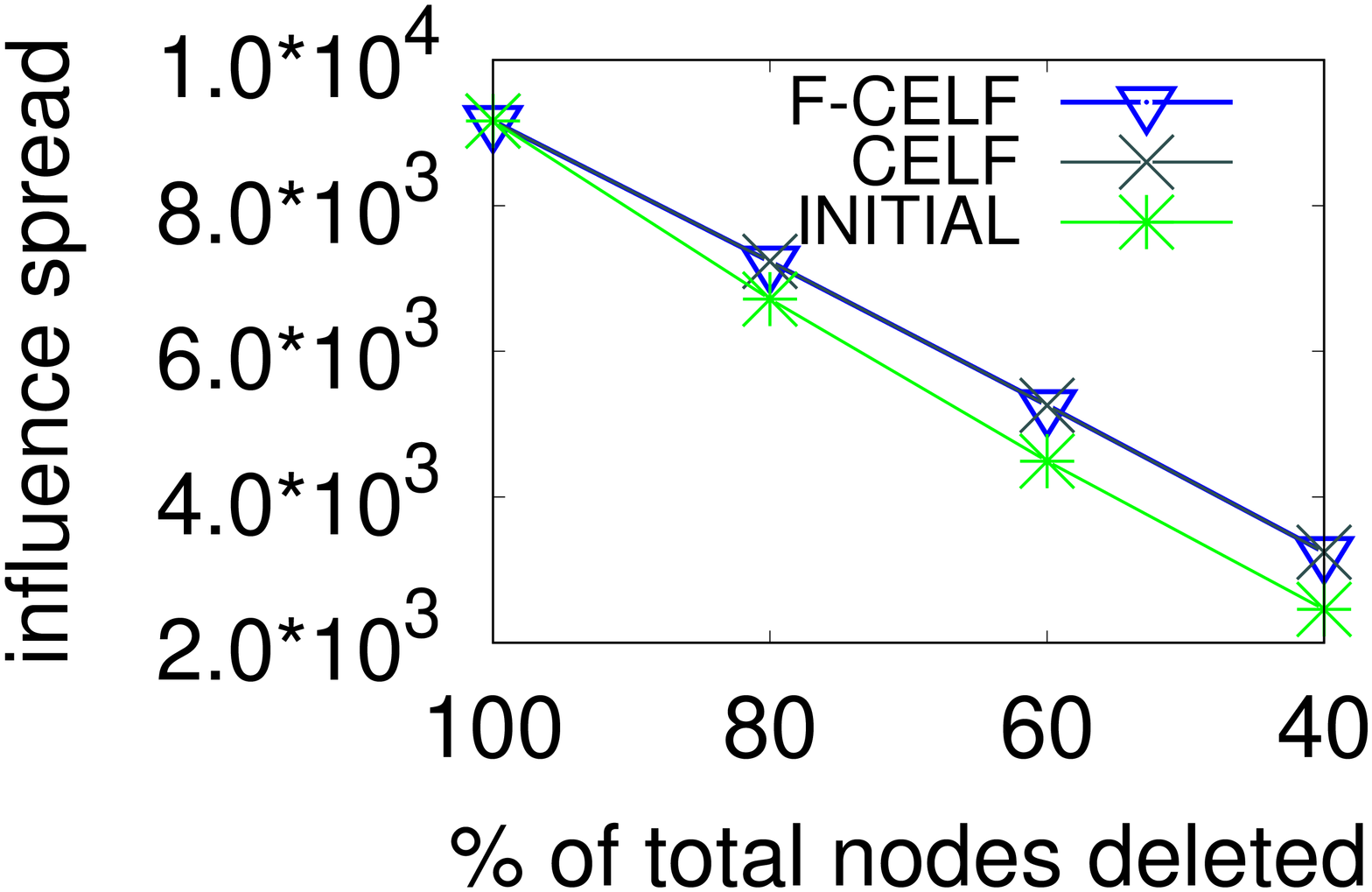}
		\label{fig:inf_sprd_epin}
	}
	\vspace{-6mm}
	\caption{\small Influence spread, seed sets are adjusted after every update}
	\label{fig:inf_sprd}
	\vspace{-5mm}
\end{figure}
\vspace{-0.5mm}
$\bullet$ {\bf Influence strength models.}
We adopt two popular edge probability models for our experiments.
{\em Those are exactly the same models used by our competitors:} {\sf UBI+} \cite{SLCHT17} {\em and} {\sf DIA} \cite{OAYK16}.
\textbf{(1) Degree Weighted Activation (\textsf{DWA}) Model.} In this model \cite{KKT03,OAYK16,SLCHT17} (also known as weighted cascade model),
the influence strength of the edge $(uv)$ is equal to $1/d_{in}(v)$, where $d_{in}(v)$ is the in-degree of
the target node $v$.
\textbf{(2) Trivalency (TV) Model.} In this model \cite{KKT03,OAYK16}, each edge is assigned with a probability, chosen uniformly at random, from $(0.1, 0.01, 0.001)$.

\vspace{-0.5mm}
$\bullet$ {\bf Competing Algorithms.}
\textbf{(1) FAMILY-CELF ({\sf F-CELF}).}
This is an implementation of our proposed {\sf N-FAMILY} framework, on top of the CELF influence maximization algorithm.
\textbf{(2) FAMILY-RR-Sketch ({\sf F-RRS}).}
This is an implementation of our proposed {\sf N-FAMILY} framework, on top of the RR-Sketch influence maximization algorithm.
\textbf{(3) DIA.}
The {\sf DIA} algorithm was proposed in \cite{OAYK16}, on top of the RR-Sketch. The method generates
all RR-sketches only once; and after every update, quickly modifies those existing sketches.
After that, {\em all seed nodes are identified from ground} using the modified sketches.
This is the key difference with our algorithm {\sf F-RRS}, since we generally need to identify
only a limited number of new seed nodes, based on the affected region due to the update.
\textbf{(4) UBI+.}
The {\sf UBI+} algorithm \cite{SLCHT17} performs greedy exchange for multiple times --- every time an old seed node is replaced
with the best possible non-seed node. If one continues such exchanges until there is no improvement, the method will guarantee
0.5-approximation. However, due to efficiency reasons, \cite{SLCHT17} limits the number of exchanges to $k$,
where $k$ is the cardinality of the seed set. An upper bounding method is used to find such best possible non-seed nodes
at every round.

\vspace{-0.5mm}
$\bullet$ {\bf Parameters Setup.}
\textbf{(1) \#Seed nodes.} We varied seed set size from 5$\sim$100 (default 30 seed nodes).
\textbf{(2) \#RR-Sketches.} Our total number of sketches are roughly bounded by $\beta(|V|+|E|)\log |V|$ as given in \cite{OAYK16}, and we varied $\beta$
from 2$\sim$512 (default $\beta=2^5=32$ \cite{OAYK16}).
\textbf{(3) Size of family.} The family size $|F_1(u)|$ of a node $u$ is decided by the parameter $\theta$, and we varied $\theta$ from 1$\sim$0.01
(default $\theta$=0.1).
\textbf{(4) \#IR to compute TIR.} We consider upto {\sf 3-IR} to compute {\sf TIR} (default upto {\sf 2-IR}).
\textbf{(5) Influence diffusion models.} We employ IC \cite{KKT03} and MIA \cite{CWW10} models for influence cascade. Bulk of our empirical
results are provided with the IC model, since this is widely-used in the literature.
\textbf{(6) \#MC samples.} We use MC simulation 10\,000 times to compute the influence spread in IC model \cite{KKT03}.

The code is implemented in Python, and the experiments are performed on a single core of a 256GB, 2.40GHz Xeon server. All results are averaged over 10 runs.
\vspace{-3mm}
\subsection{Single Update Results}
\label{sec:single_update_exp}
\vspace{-1mm}
First, we show results for single update queries related to edge addition, edge deletion, node addition, and node deletion. We note that adding an edge $uv$ can also be considered as an increase in the edge probability from $0$ to $P_e(uv)$. Analogously, deleting an edge can be regarded as a decrease in edge probability. Moreover, for the DWA edge influence model, when an edge is added or deleted, the probabilities of multiple adjacent edges are updated (since, inversely proportional to node degree).
\textbf{(1) Edge addition.} We start with initial 40\% of the edges in the graph data, and then add all the remaining edges as dynamic updates.
We demonstrate our results with the {\sf Digg} dataset and the \textsf{DWA} edge influence model (Figure~\ref{fig:edgeadd_digg}).
\textbf{(2) Edge deletion.} We delete the last 60\% of edges from the graph as update operations. We use the {\sf Slashdot} dataset, with \textsf{TV} model, for showing our results (Figure~\ref{fig:edgedel_slash}).
\textbf{(3) Node addition.}
We start with the first $40$\% of nodes and all their edges in the dataset. We next added the remaining nodes sequentially, along with their associated edges. We present our results over {\sf Epinions}, along with the \textsf{TV} model (Figure~\ref{fig:nodeadd_epin}).
\textbf{(4) Node deletion.}
We delete the last $20$\% of nodes, with all their edges from the graph. We use our largest dataset {\sf Flickr} and the \textsf{DWA} model for demonstration (Figure~\ref{fig:nodedel_flickr}).

For the aforesaid operations, we adjust the seed set after every update, since one does not know apriori when the seed set actually changes, and hence, it can be learnt only after updating them.

\vspace{-0.5mm}
\spara{Efficiency.}
In Figure~\ref{fig:eff_ic}, we present the running time to dynamically adjust the top-$k$ seed nodes, under the IC influence cascade model.
We find that {\sf F-CELF} and {\sf F-RRS} are always faster than {\sf UBI+} and {\sf DIA}, respectively, by 1$\sim$2 orders of magnitude.
As an example, for node addition over {\sf Epinions} in Figure~\ref{fig:nodeadd_epin}, the time taken by {\sf F-CELF}
is only $2\times10^6$ sec for about $80$K node additions (i.e., 24.58 sec/node add). In comparison, {\sf UBI+} takes around
$8\times10^7$ sec (i.e., 1111.21 sec/ node add). Our {\sf F-RRS} algorithm requires about $4\times10^5$ secs (i.e., 5.31 sec/ node add),
and {\sf DIA} takes $10\times10^6$ sec (i.e., 134.68 sec/node add). {\em These results clearly demonstrate the efficiency improvements by our methods}.

We also note that sketch-based methods are relatively slower (i.e., {\sf F-RRS} vs. {\sf F-CELF}, and {\sf DIA} vs. {\sf UBI+})
in smaller graphs (e.g., {\em Digg} and {\em Slashdot}). This is due to the overhead of updating sketches after graph updates.
On the contrary, in our larger datasets, {\sf Epinions} and {\sf Flickr}, the benefit of sketches is more evident as opposed to
MC-simulation based techniques. In fact, both {\sf F-CELF} and {\sf UBI+} are very slow for our largest {\em Flickr} dataset (see Table~\ref{tab:result_summary}); hence, we only show {\sf F-RRS} and {\sf DIA} for {\sf Flickr} in Figure~\ref{fig:nodedel_flickr}.

Additionally, in Figure~\ref{fig:eff_MIA}, we show the efficiency of our method under the MIA model of influence spread.
Since it is non-trivial to adapt {\sf UBI+} and {\sf DIA} for the MIA model, we compare our algorithm {\sf F-CELF}
with {\sf CELF} \cite{LKGFVG07} in these experiments. For demonstration, we consider
{\sf Slashdot} and {\sf Epinions}, together with node addition and deletion, respectively.
It can be observed from Figure~\ref{fig:eff_MIA} that {\sf F-CELF} is about 2 orders of magnitude
faster than {\sf CELF}. {\em These results illustrate the generality and effectiveness
of our approach under difference influence cascading models}.
\begin{table}
	\vspace{1mm}
	\caption{\small Memory consumed by different algorithms}
	\scriptsize
	\vspace{-3mm}
	\centering
	\begin{tabular} { |l||c|c|c|l| }
		\hline
		{\textsf{Algorithms}} & {\textsf{Digg}} & {\textsf{Slashdot}} & {\textsf{Epinions}}  &  {\textsf{Flickr}} \\
		\hline
		{\sf F-CELF}, {\sf UBI+} & 0.22 GB & 0.32 GB & 1.03 GB & 31.55 GB\\
		\hline \hline
		{\sf F-RRS}, {\sf DIA} & 3.83 GB & 5.89 GB & 25.87 GB & 142.89 GB\\
		\hline
	\end{tabular}
	\label{tab:memory}
	\vspace{-5mm}
\end{table}

\vspace{-0.5mm}
\spara{Influence spread.}
We report the influence spread with the updated seed set for both IC (Figure~\ref{fig:inf_sprd_digg}) and MIA models (Figure~\ref{fig:inf_sprd_epin}).
It can be observed that the competing algorithms, i.e., {\sf F-CELF}, {\sf F-RRS}, {\sf UBI+}, and {\sf DIA}  achieve similar influence
spreads with their updated seed nodes. Furthermore, we also show by {\sf INITIAL} the influence spread obtained by the old seed set in the modified graph. We find that {\sf INITIAL} achieves significantly less influence spread, especially with more graph updates.
{\em These results demonstrate the usefulness of dynamic IM techniques in general, and also the effectiveness of our algorithm in terms of influence coverage}.

\vspace{-0.5mm}
\spara{Memory usage.}
We show the memory used by all algorithms in Table~\ref{tab:memory}. We find that MC-sampling based algorithms (i.e., {\sf F-CELF} and {\sf UBI+}) take similar amount of memory, whereas sketch-based techniques (i.e., {\sf F-RRS} and {\sf DIA}) also have comparable memory usage. {\em Our results illustrate that the proposed methods, {\sf F-CELF} and {\sf F-RRS} improve the updating time of the top-$k$ influencers by 1$\sim$2 orders of magnitude, compared to state-of-the-art algorithms, while ensuring similar memory usage and influence spreads.}
\vspace{-3mm}
\subsection{Batch Update Results}
\label{sec:batch_update_exp}
\vspace{-1mm}
We demonstrate batch updates with a sliding window model as used in \cite{SLCHT17}. In this model, initially we consider the edges present in between $0$ to $W$ units of time (length of window) and compute the seed set. Next, we slide the window to $L$ units of time. The edges present in between $L$ and $W+L$ are considered as the updated data, and our goal is to adjust the seed set based on the updated data. We delete the edges from $0$ to $L$ and add the edges from $W$ to $W+L$. We continue sliding the window until we complete the whole data.

We conducted this experiment using the {\em Twitter} dataset downloaded from https://snap.stanford.edu/data/.
The dataset is extracted from the tweets posted between 01-JUL-2012 to 07-JUL-2012, which is during the
announcement of the Higgs-Boson particle. This dataset contains $304\,199$ nodes and $555\,481$ edges.
Probability of an edge $uv$ is given by the formula $1-e^\frac{-f}{k}$, where $f$ is the total number of
edges appeared in the window, and $k$ is the constant. We present our experimental results by
varying $W$ from 30 mins to 6 hrs and $L$ from 1 sec to 2 mins. We set the value of $k$ as $5$.
On an average, $1.8$ updates appear per second. Since the number of edges in a window is small,
we avoid showing results with {\sf F-RRS}. This is because {\sf F-CELF} performs much better
on smaller datasets. From the experimental results in Figure~\ref{fig:batch_ic}, we find that {\em {\sf F-CELF} is faster than both {\sf UBI+} and {\sf DIA} upto three orders of magnitude}.
\begin{figure}[t!]
	\vspace{-1mm}
	\centering
	\subfigure[{\em  Run time to adjust seed set, varying $L$, $W$ = 1 hour}] {
		\includegraphics[scale=0.14, angle=270]{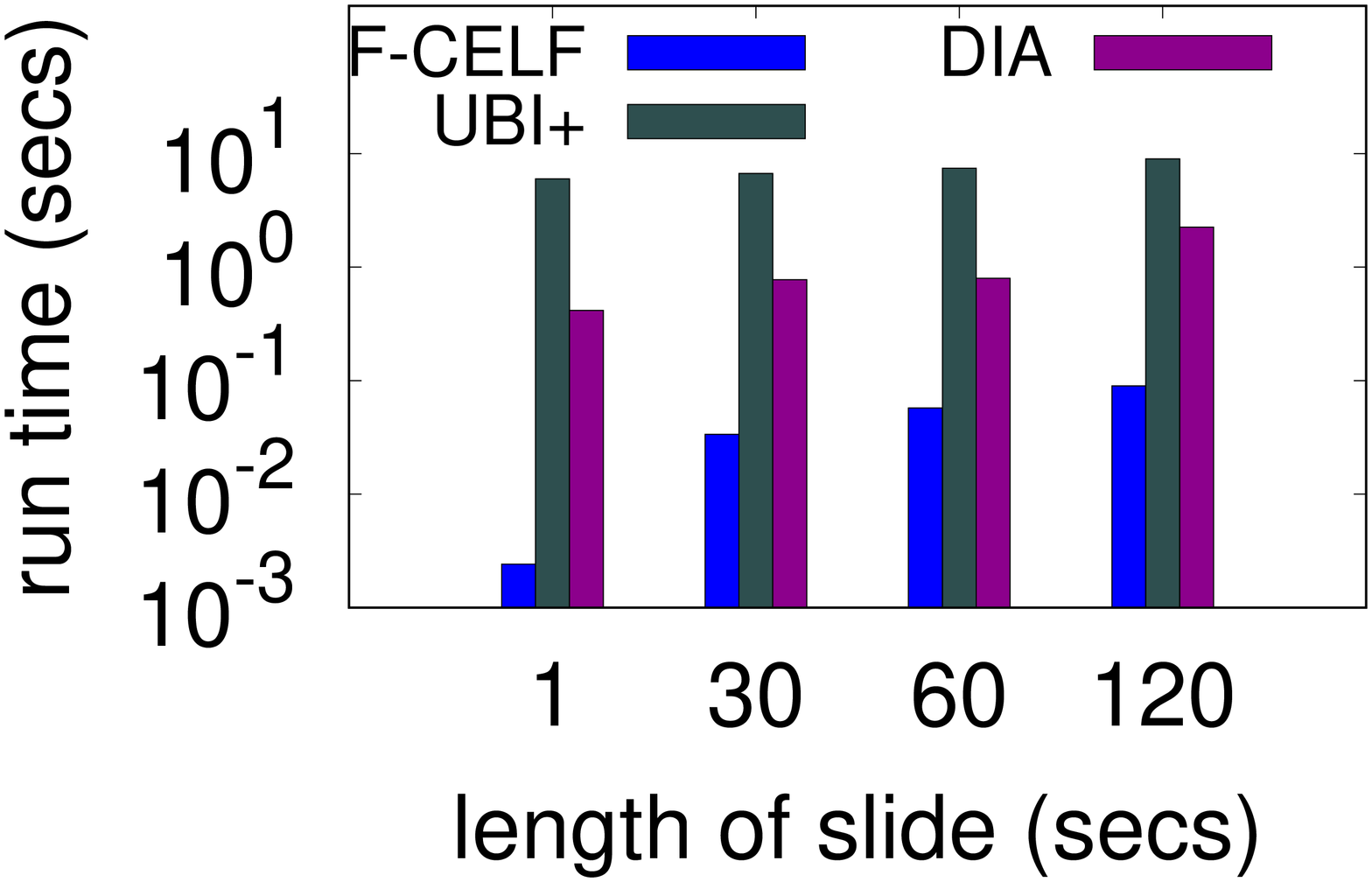}
		\label{fig:var_L}
	}
  	\subfigure[{\em Run time to adjust seed set, varying $W$, $L$ = 60 secs}]  {
		\includegraphics[scale=0.14, angle=270]{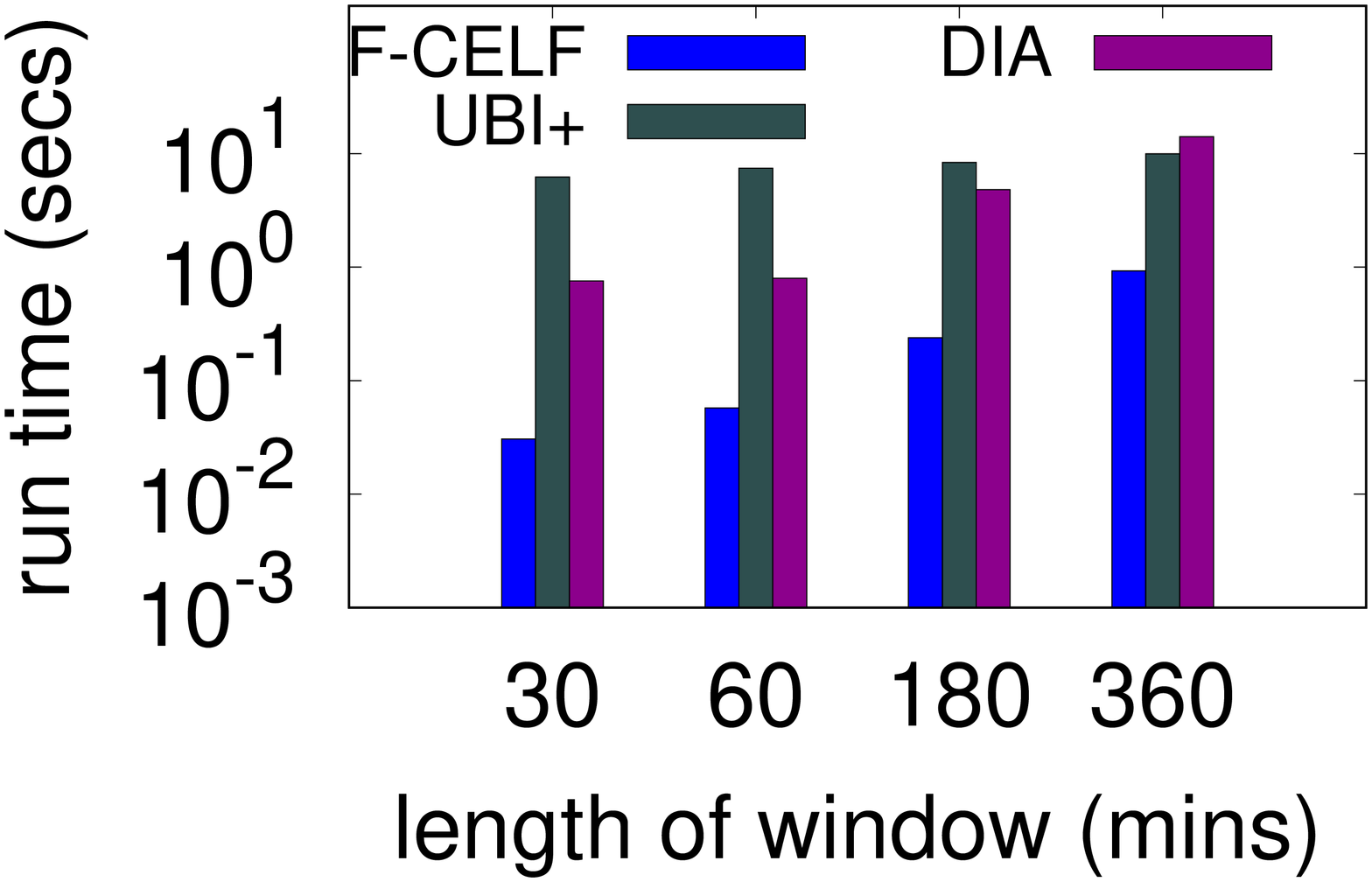}
		\label{fig:var_w}
	}
	\vspace{-6mm}
	\caption{\small Impacts of varying batch sizes, sliding window model, {\em Twitter}, IC influence prop., seed sets are adjusted after every slide}
	\label{fig:batch_ic}
	\vspace{-6mm}
\end{figure}

\vspace{-3mm}
\subsection{Sensitivity Analysis}
\label{sec:sens_anl}
\vspace{-1mm}
In these experiments, we vary the parameters of our algorithms. For demonstration,
we update the last 40$\%$ nodes in a dataset, and report the average time taken to
re-adjust the seed set per update operation, with the {\sf F-RRS} algorithm.

With increase in the size of {\sf TIR}, number of seed nodes that get infected may increase. For a given update, size of {\sf TIR} depends of two factors: $\theta$ (with decrease in $\theta$, size of family increases: Figure~\ref{fig:theta_family_size}) and $\#{\sf IR}$ to compute {\sf TIR}. Hence, we vary $\theta$ (Figure~\ref{fig:theta}) and $\#{\sf IR}$s (Figure~\ref{fig:family}) for node deletion in {\em Epinions}. We find that by selecting $\theta=10^{-1}$, influence spread increases by around $8.4$\% compared to that of $\theta=10^{0}$, and there is no significant increase in influence spread for even smaller $\theta$. Similarly, for increase in $\# {\sf IR}$ almost saturates at $\#${\sf IR}=$2$.
However, the efficiency of the algorithm decreases almost linearly with decrease in $\theta$ (Figure~\ref{fig:eff_theta_epin}) and increase in $\# {\sf IR}$s. Hence, by considering a trade off between efficiency and influence coverage we select $\theta = 10^{-1}$ and $\#${\sf 2-IR} to to compute {\sf TIR}.
\begin{figure}[t!]
    \vspace{-1mm}
	\centering
	\subfigure[{\em Inf. Spread}]  {
		\includegraphics[scale=0.09, angle=270]{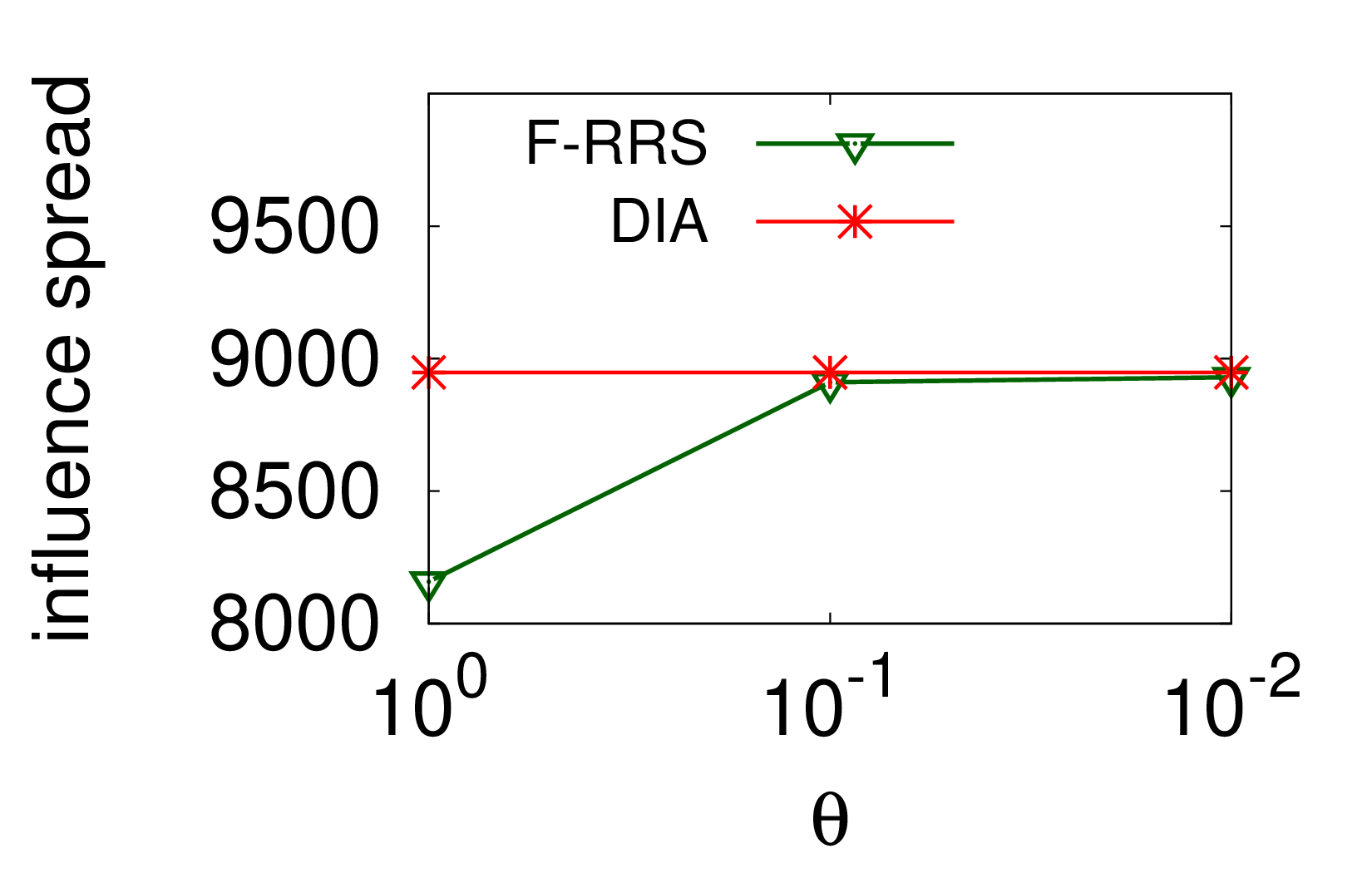}
		\label{fig:acc_theta_epin}
	}
	\subfigure[{\em Run time to adjust seeds}] {
		\includegraphics[scale=0.09, angle=270]{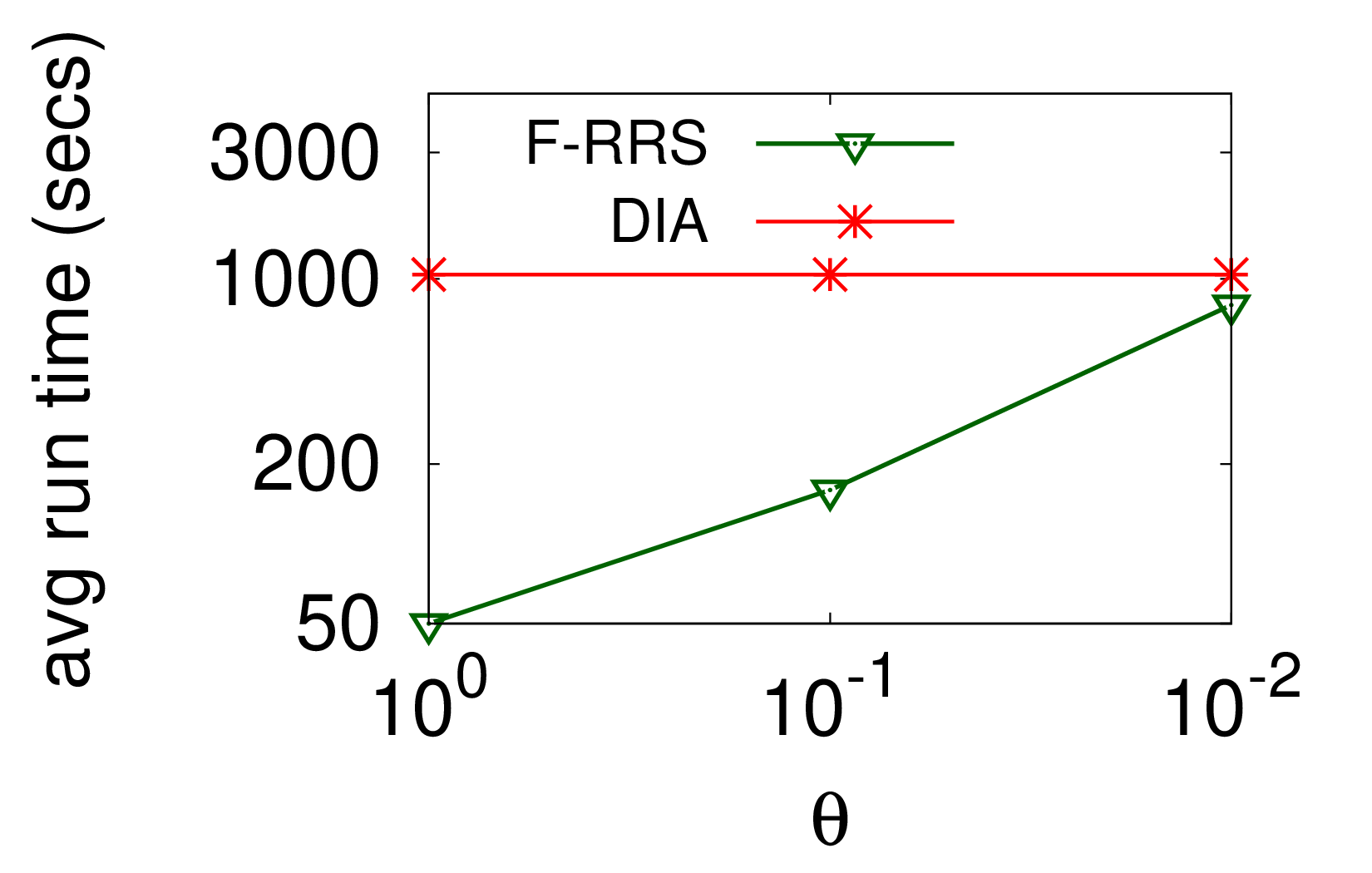}
		\label{fig:eff_theta_epin}
	}
	\subfigure[{\em Avg. family size w/ $\theta$}]  {
			\includegraphics[scale=0.09, angle=270]{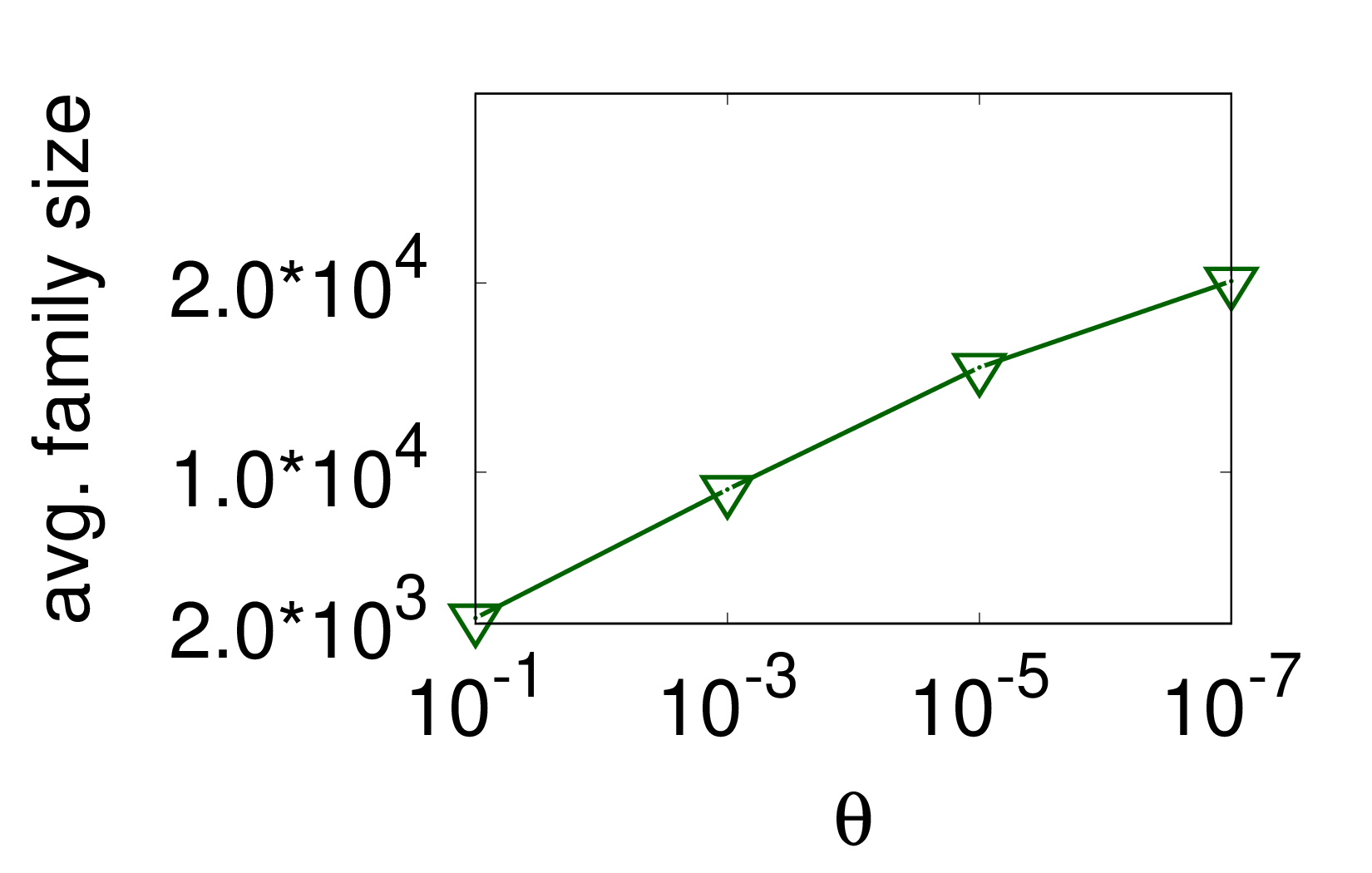}
			\label{fig:theta_family_size}
	}
	\vspace{-6mm}
	\caption{\small Impacts of $\theta$, node del., {\em Epinions} (\textsf{DWA}), IC model}
	\label{fig:theta}
	\vspace{-5mm}
\end{figure}
\begin{figure}[t!]
	\vspace{-2mm}
	\centering
	\subfigure[{\em Inf. spread}]  {
		\includegraphics[scale=0.14, angle=270]{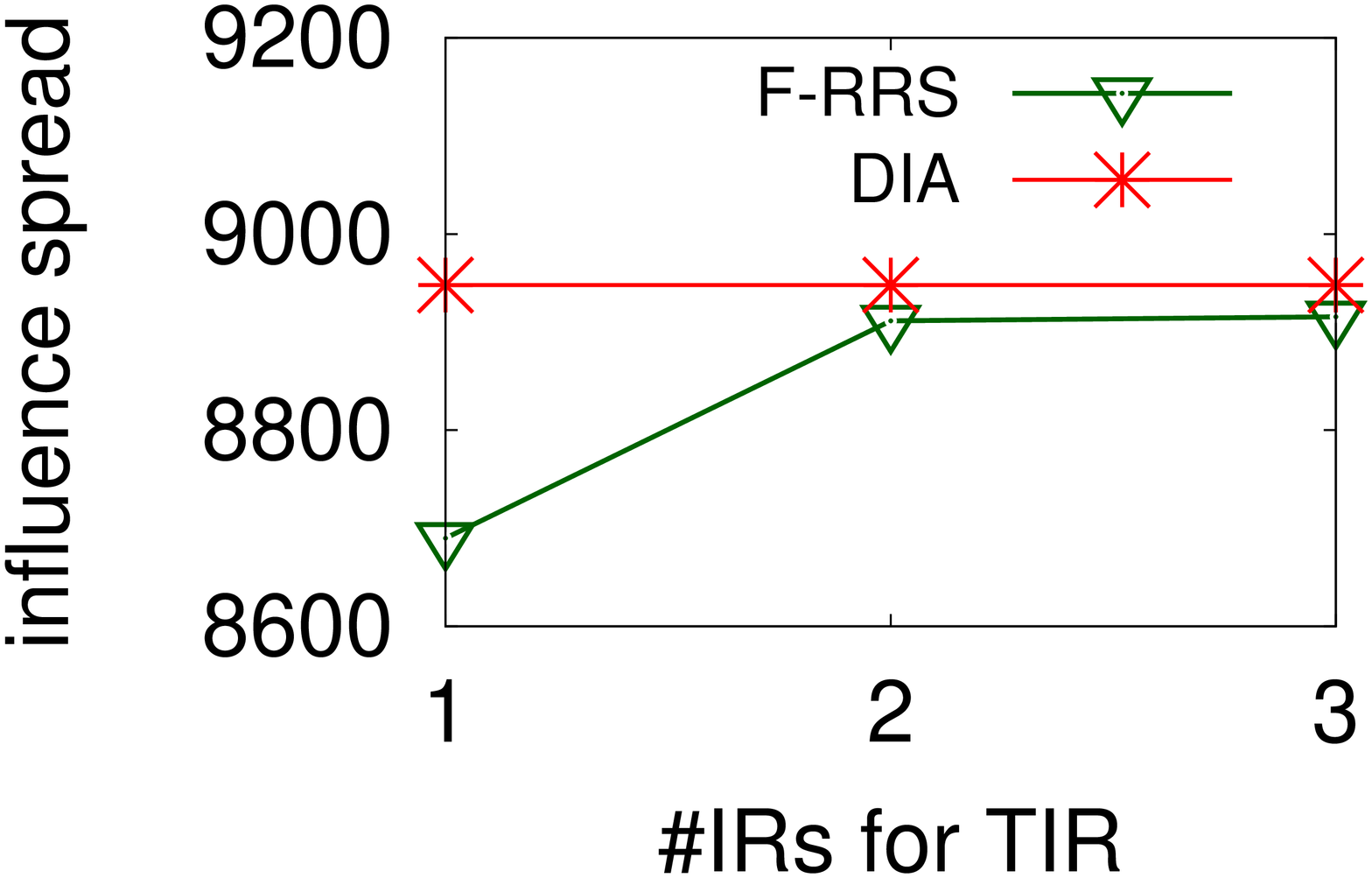}
		\label{fig:acc_family}
	}
	\subfigure[{\em Run time to adjust seed set}] {
		\includegraphics[scale=0.14, angle=270]{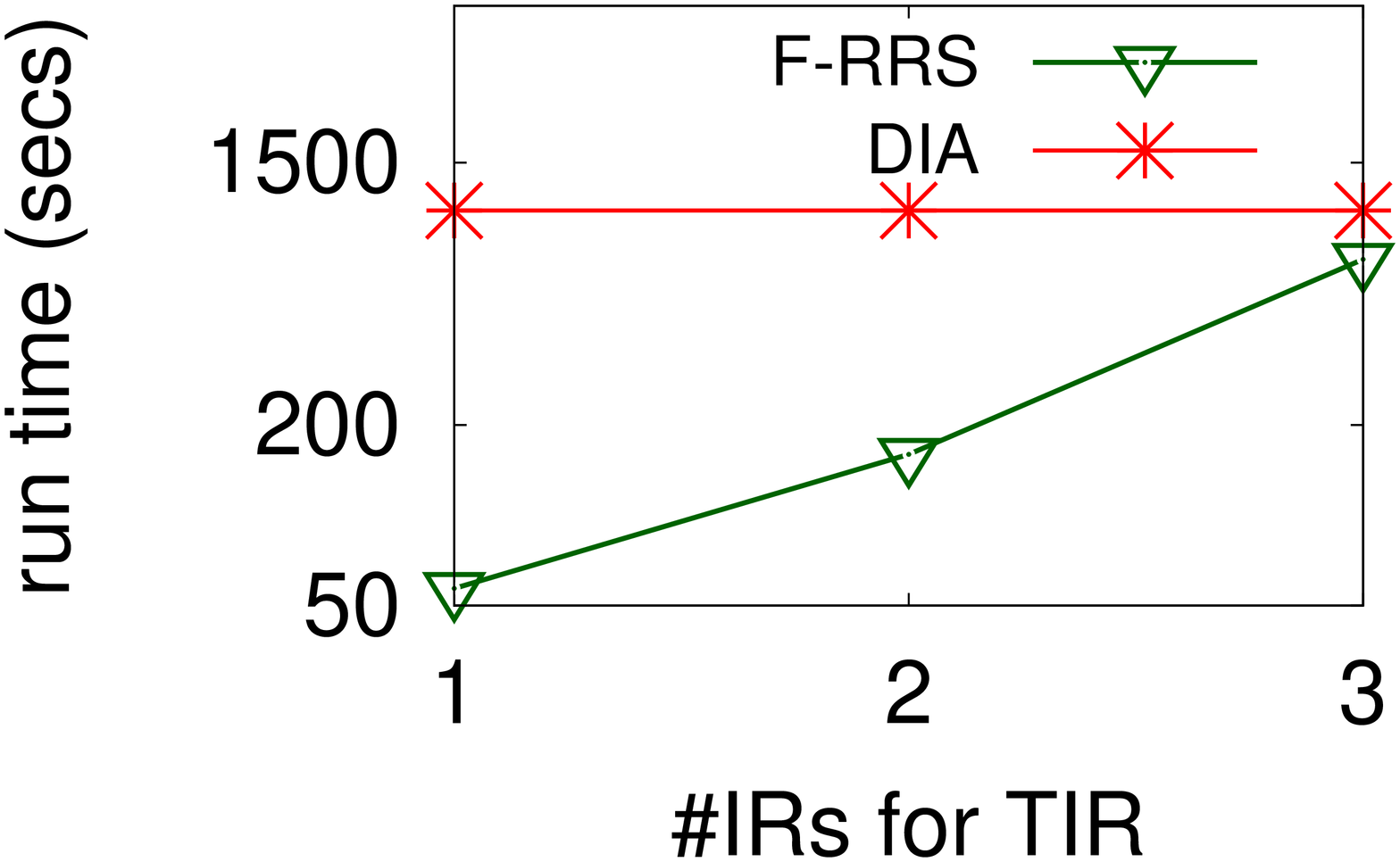}
		\label{fig:eff_family}
	}
	\vspace{-6mm}
	\caption{\small Impacts of \#IRs, node del., {\em Epinions} (\textsf{DWA}), IC model}
	\label{fig:family}
	\vspace{-5mm}
\end{figure}
\begin{figure}[t!]
	\vspace{-2mm}
	\centering
	\subfigure[{\em Run time to adjust seeds, node add., {\em Epinions}}] {
		\includegraphics[scale=0.14, angle=270]{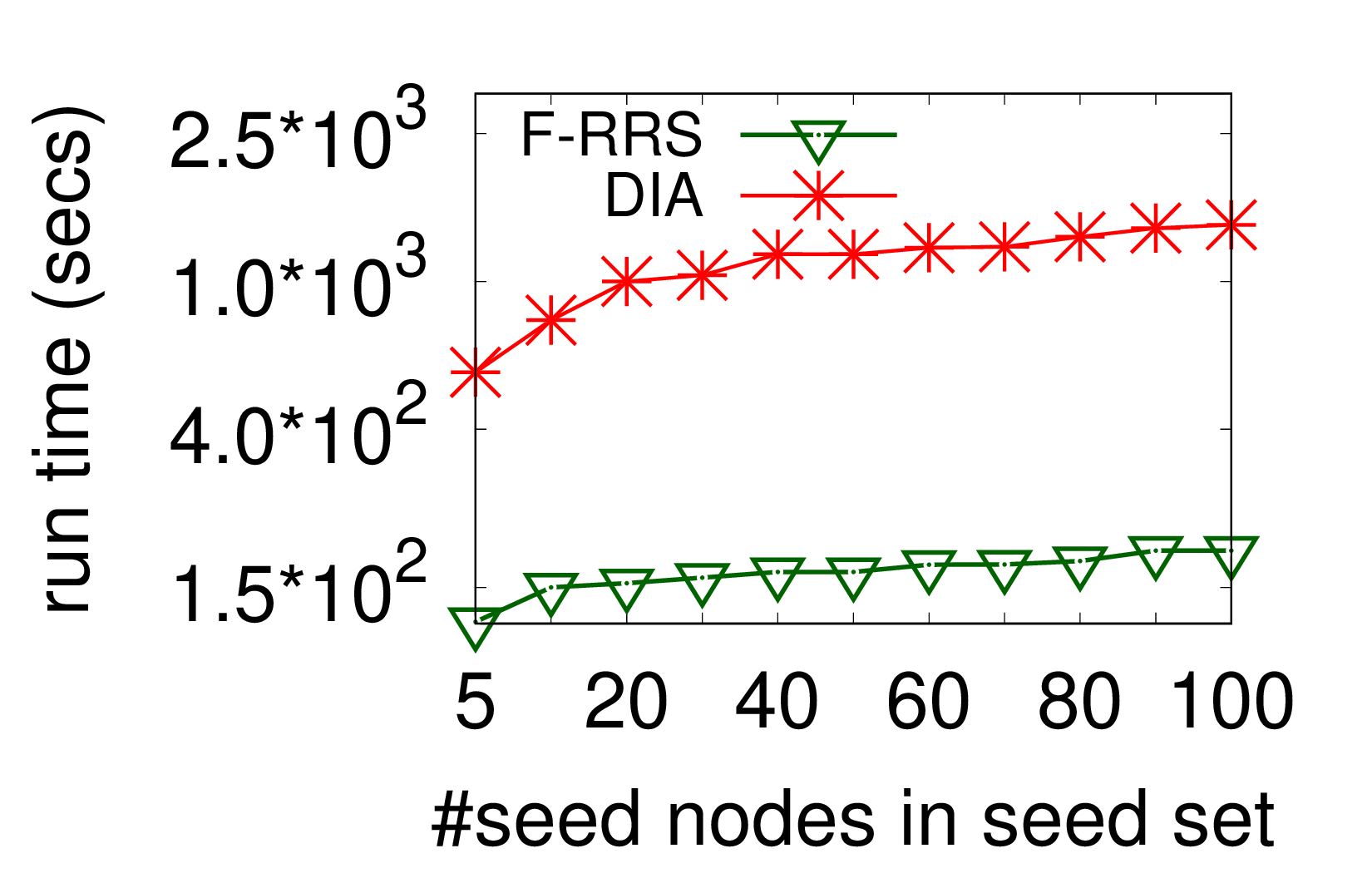}
		\label{fig:seed}
	}
	\subfigure[{\em Inf. spread w/ varying $\beta$, \newline {\em Digg}}]  {
		\includegraphics[scale=0.14, angle=270]{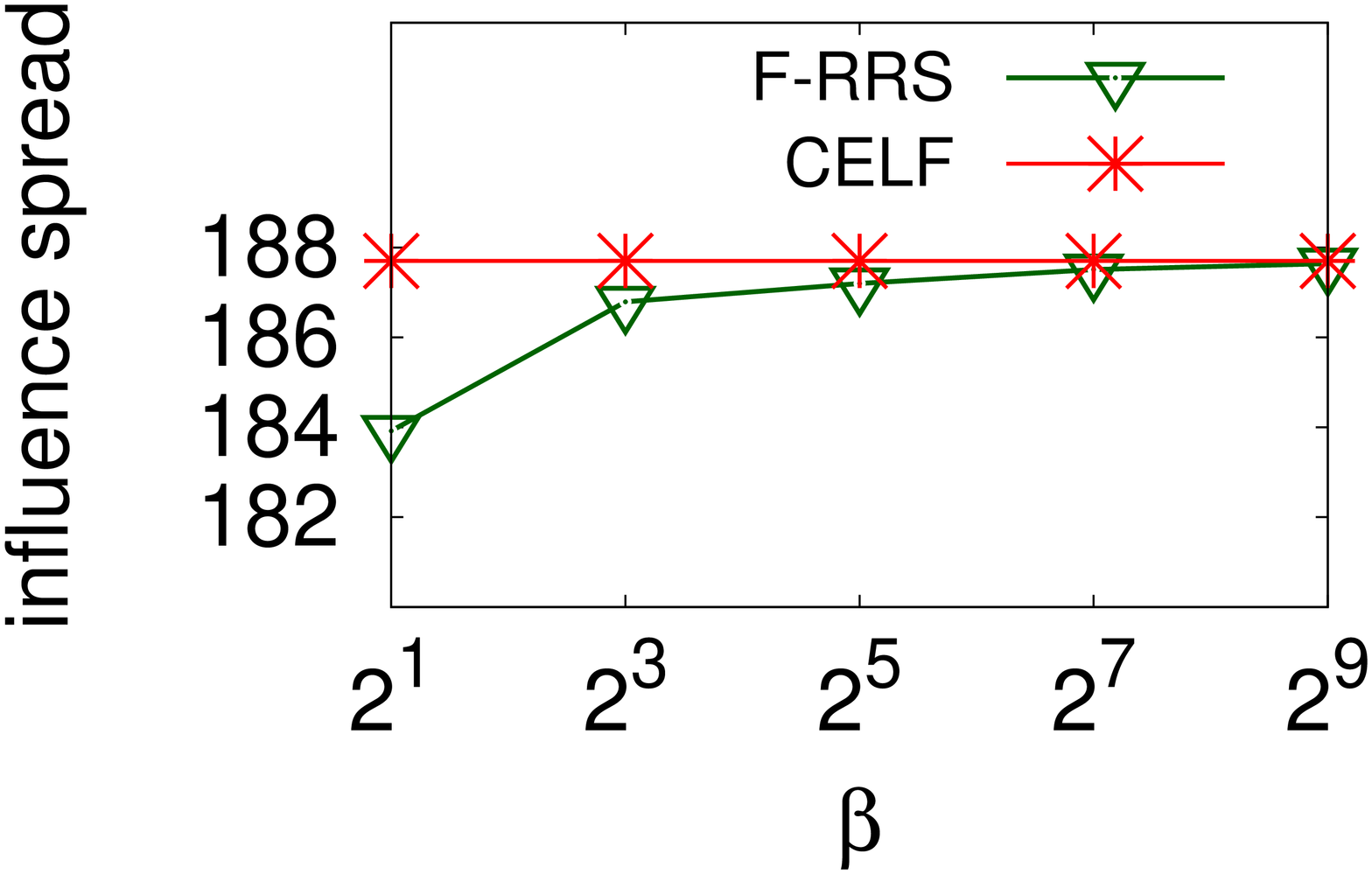}
		\label{fig:beta_digg}
	}
	\vspace{-6mm}
	\caption{\small Impacts of varying \#seeds and $\beta$, IC model}
	\label{fig:seed_beta}
	\vspace{-6mm}
\end{figure}

In Figure~\ref{fig:seed}, we show the efficiency with \textbf{varying seed set size} from $5$ to $100$. It can be observed that even for the seed set of size $100$, {\sf F-RRS} is faster than {\sf DIA} by more than an order of magnitude. {\em This demonstrates that our technique is scalable for large seed set sizes.} For sketch-based methods, choosing the optimal $\beta$ is very important. In Figure~\ref{fig:beta_digg}, we show the influence coverage of the {\sf F-RRS} with \textbf{varying $\beta$} from $2$ to $512$. We compare the influence spread with that of {\sf CELF}. We find that with increase in $\beta$, influence coverage  initially increases, and gets saturated at $\beta=32$. Hence, we set $\beta=32$ in our experiments.

%% file: conclusion.tex
\vspace{-1mm}
\section{Conclusions}
\label{sec:conclusions}
We developed a generalized, local updating framework for efficiently adjusting the top-$k$ influencers
in an evolving network. Our method iteratively identifies {\em only} the
affected seed nodes due to dynamic updates in the influence graph, and then replaces
them with more suitable ones. Our solution can be applied to a variety of
information propagation models and influence maximization techniques.
Our algorithm, {\sf N-Family} ensures $(1-\frac{1}{e})$ approximation guarantee with the MIA
influence cascade model, and works well for localized batch updates. Based on a detailed empirical
analysis over several real-world, dynamic, and large-scale networks, {\sf N-Family} improves
the updating time of the top-$k$ influencers by 1$\sim$2 orders of magnitude, compared to state-of-the-art algorithms, while ensuring
similar memory usage and influence spreads.

%% file: appendix.tex
\section{Proof of Lemma 3}
\label{sec:proof_lemma3}
According to Eq.~\ref{eq:marginal}, the marginal gain of $u$ with respect to $S$ is given as:
\vspace{-1mm}
\begin{align}
	MG(S,u) &= \sigma{(S\cup \{u\})} - \sigma{(S)}  \nonumber \\
	&= \sum_{w\in V}{pp(S\cup \{u\},w)} - \sum_{w\in V}{pp(S,w)}  \nonumber
\end{align}
\begin{align}
	& = \hspace{-4mm} \sum_{w\in V\setminus F_1(s)}\hspace{-4mm}{pp(S\cup \{u\},w)} + \sum_{w\in F_1(s)}{pp(S\cup \{u\},w)} - \nonumber \\
	& \sum_{w\in V\setminus F_1(s)}\hspace{-4mm}{pp(S,w)} - \sum_{w\in F_1(s)}{pp(S,w)} \label{eq:rem_mg_1}
\end{align}
As $u \notin F_2(s)$, the influence of $u$ on any node in $F_1(s)$ is $0$. Hence, Equation~\ref{eq:rem_mg_1} can be written as:
\vspace{-1mm}
\begin{align}
	MG(S,u) & = \hspace{-2mm} \sum_{w\in V\setminus F_1(s)\hspace{-4mm}}{pp(S\cup \{u\},w)} + \sum_{w\in F_1(s)}{pp(S,w)} - \nonumber \\
	& \sum_{w\in V\setminus F_1(s)}\hspace{-4mm}{pp(S,w)} - \sum_{w\in F_1(s)}{pp(S,w)} \nonumber \\
	&  = \sum_{w\in V\setminus F_1(s)}\hspace{-4mm}{pp(S\cup \{u\},w)} - \sum_{w\in V\setminus F_1(s)}\hspace{-4mm}{pp(S,w)} \label{eq:rem_mg_2}
\end{align}
Now, the removed seed node $s$ cannot influence any node outside $F_1(s)$. Hence, Equation~\ref{eq:rem_mg_2}  can be written as:
\vspace{-1mm}
\begin{align}
	MG(S,u) & = \sum_{w\in V\setminus F_1(s)}\hspace{-4mm}{pp({S\setminus \{s\}} \cup \{u\},w)} - \hspace{-4mm} \sum_{w\in V\setminus F_1(s)}\hspace{-4mm}{pp({S \setminus \{s\}},w)} \nonumber \\
	& = \hspace{-4mm} \sum_{w\in V\setminus F_1(s)} \hspace{-4mm}{pp({S\setminus \{s\}} \cup \{u\},w)} + \hspace{-2mm} \sum_{w\in F_1(s)} \hspace{-2mm}{pp({S\setminus \{s\}} \cup \{u\},w)} \nonumber \\
	& - \hspace{-4mm} \sum_{w\in V\setminus F_1(s)}\hspace{-4mm}{pp({S \setminus \{s\}},w)} - \sum_{w\in F_1(s)}\hspace{-2mm}{pp({S\setminus \{s\}} \cup \{u\},w)} \label{eq:rem_mg_3}
\end{align}
As influence of $u$ on any node in $F_1(s)$ is $0$, Equation~\ref{eq:rem_mg_3} can be written as:
\vspace{-1mm}
\begin{align}
	MG(S,u) & = \hspace{-4mm} \sum_{w\in V\setminus F_1(s)}\hspace{-4mm}{pp({S\setminus \{s\}} \cup \{u\},w)} +  \hspace{-4mm} \sum_{w\in F_1(s)}\hspace{-2mm}{pp({S\setminus \{s\}} \cup \{u\},w)} \nonumber \\
	& - \hspace{-4mm} \sum_{w\in V\setminus F_1(s)}\hspace{-4mm}{pp({S \setminus \{s\}},w)} - \sum_{w\in F_1(s)}{pp({S\setminus \{s\}},w)} \nonumber \\
	& = \sum_{w\in V}{pp({S\setminus \{s\}}\cup \{u\},w)} - \sum_{w\in V}{pp({S\setminus \{s\}},w)}  \nonumber \\
	&= \sigma{({S\setminus \{s\}} \cup \{u\})} - \sigma{(S\setminus \{s\})}  \nonumber \\
	& = MG(S\setminus \{s\}, u)
\end{align}
Hence, the lemma.

\section{Proof of Lemma 4}
\label{sec:proof_lemma4}
Consider a node $w$ outside $MIIA(u,\theta)$ in the original graph $\mathcal{G}(V,E,P)$,
which means $w$ cannot activate $u$ with a minimum strength of $\theta$ through $MIP(w,u)$.
Then, the strength at which $w$ activates $v$ through $u$ in the updated graph is:
$pp(w,u)\times P_{u,v}$. Since, $pp(w,u) < \theta$,
we have: $pp(w,u)\times P_{u,v}< \theta$.
Thus, adding the edge $uv$ does not change the expected influence spread
of $w$, based on the MIA model. Hence, the lemma follows.

\section{Proof of Lemma 5 and 6}
\label{sec:proof_lemma56}
A seed node can influence only the nodes present in its family according to the MIA model. There is no node present in {\sf TIR} which belongs to the family of any seed node outside {\sf TIR}. This is because any uninfected seed node is more than 2-Family away from any node present in {\sf TIR} (This is how we terminate infection propagation). Hence, both Lemma 5 and 6 follow.

\section{Proof of Performance \\ Guarantee under MIA Model}
\label{sec:theo_acc}
We show that the top-$k$ seed nodes reported by our {\sf N-Family} method (Algorithm~\ref{alg:prop}) are the same as the top-$k$ seed nodes obtained by running the Greedy
on the updated graph under the MIA model. Since, the Greedy algorithm provides the approximation guarantee of $1-\frac{1}{e}$, our {\sf N-Family} also provides the same approximation guarantee.
The proof is as follows.

As described in Section~\ref{sec:update_algo}, after identifying the {\sf TIR} using Equation~\ref{eq:TIR}, we compute $S_{rem}$ (=$S\setminus {\sf TIR}$), influence spreads of all nodes $u \in {\sf TIR}$, and update the priority queue.

Now, we continue with computing the $k-k'$ new seed nodes over the updated graph, and $S_{new}$ is new seed set (of size $k$) found in this manner. Note that before we begin computing new seed nodes, $S_{new}$ contains the $k'$ seed nodes present in $S_{rem}$, and then
new nodes are added in an iterative manner. Clearly, $S_{new}^{k'}$ is same as $S_{rem}$. We consider $s_{n}^{k'+i}$ as
the seed node computed by Greedy in the $i^{th}$ iteration, where $i\le k-k'$.
Due to Greedy algorithm,
\begin{align}\label{eq:snew_greedy}
MG(S_{new}^{k-1},s_n^k) \ge MG(S_{new}^{k-1},u), \qquad \forall u\in V\setminus S_{new}
\end{align}

Next, we sort all seeds in $S_{new}$ according to the greedy inclusion order,
and the sorted seed set is denoted as $S_{order}$. Note that seed nodes present in $S_{new}$ and $S_{order}$ are same, but their
order could be different. At this stage, the important observations are as follows.

After computing $S_{order}$, and assuming $w$ the top-most node in the priority queue, we will have two mutually exclusive cases:

\spara{Case 1:} $MG(S_{order}^{k-1},s_o^k) \ge MG(S_{order},w)$\\
\spara{Case 2:} $MG(S_{order}^{k-1},s_o^k) < MG(S_{order},w)$

If we end up with Case-1, we terminate our algorithm and report $S_{order}$ as the set of new seed nodes,
which would be same as the ones computed by the Greedy algorithm on the updated graph (we shall prove this soon).
However, if we arrive at Case-2, we do iterative seed replacements until we achieve Case-1
(we prove that by iterative seed replacements for at most $k'$ times, we reach Case-1).

Moreover, there are two more mutually exclusive cases which can be derived from the following lemma.
\begin{lma}\label{lma:last_seed}
	The last seed node, $s_o^k$ present in $S_{order}$ can be either $s_r^{k'}$  (i.e., the
	last seed node in $S_{rem}$) or $s_n^k$ (i.e., the last seed node in $S_{new}$).
\end{lma}
\begin{proof}
Since we compute new seed nodes using Greedy algorithm for $S_{new}$, $MG(S_{new}^{k-1},s_n^k) \le MG(S_{new}^{k'+i-1},s_n^{k'+i})$ for $0<i\le(k-k')$ (Inequality~\ref{eq:greedy_marginal}), and $MG(S_{rem}^{k'-1},s_r^{k'}) \le MG(S_{rem}^{k'-l},s_r^{k'-l+1})$ (Inequality~\ref{eq:order_rem})
	where $2\le l\le k'$. In other words, $s_{r}^{k'}$ and $s_n^k$ provide least marginal gains compared to other nodes in $S_{rem}$ and in $S_{new}/S_{rem}$, respectively. Hence, $s_o^k$ can be either $s_n^k$ or $s_{r}^{k'}$.
\end{proof}
Therefore, the two mutually exclusive cases are:

\spara{Case A:} $s_{o}^{k} = s_n^{k}$ (i.e., the last seed node in $S_{new}$)\\
\spara{Case B:} $s_{o}^{k} = s_r^{k'}$ (i.e., the last seed node in $S_{rem}$)

Now, we will show that the seed nodes obtained after reaching Case-1, i.e., when $MG(S_{order}^{k-1},s_o^k) \ge MG(S_{order},w)$,
and under both Case-A and Case-B, i.e., $s_{o}^{k} = s_n^{k}$ and $s_{o}^{k} = s_r^{k'}$, are exactly same as the seed nodes produced by Greedy algorithm on the updated graph. For the seed set computed by the Greedy algorithm on the updated graph, the following inequality must hold.
\begin{align}\label{eq:greedy_req}
	MG(S_{order}^{k-1},s_o^k) \ge MG(S_{order}^{k-1},v), \qquad \forall v\in V\setminus S_{order}
\end{align}
Hence, we will prove that for Case-1, Inequality~\ref{eq:greedy_req} is true in both Case-A and Case-B.

First, we will show for Case-1 ($MG(S_{order}^{k-1},s_o^k) \ge MG(S_{order},w)$) and Case-A ($s_{o}^{k} = s_n^{k})$.
\begin{lma}\label{lma:subcase1}
	If $MG(S_{order}^{k-1},s_o^k) \ge MG(S_{order},w)$ (Case-1), where $w$ is the top-most node in the priority queue, and $s_o^k = s_n^k$ (Case-A), then $S_{order}$ is the set of seed nodes computed by Greedy on the updated graph, i.e., Inequality~\ref{eq:greedy_req} holds.
\end{lma}
\begin{proof}
Given $s_{o}^k = s_n^k$. Moreover, seed nodes present in $S_{order}$ and $S_{new}$ are same. Hence,
\begin{align}\label{eq:sok_snk}
MG(S_{order}^{k-1},s_o^k) = MG(S_{new}^{k-1},s_n^k)
\end{align}
Next, by combining Equation~\ref{eq:sok_snk} and Inequality~\ref{eq:snew_greedy} we get, $MG(S_{order}^{k-1}$ $,s_o^k) \ge MG(S_{order}^{k-1},u)$,
for all $u \in V\setminus S_{order}$.
Hence, the lemma.
\end{proof}
Now, to prove the theoretical guarantee for Case-1 and Case-B, the following Lemma is very important.
\begin{lma}\label{lma:seedtir}
	If $s_o^k = s_r^{k'}$, then
	\begin{enumerate}
		\item All new seed nodes computed belong to {\sf TIR}, i.e., \\
		$S_{order}\setminus S_{rem} \in \text{\sf TIR}$.
		\item $MG(S_{order}^{k-1},s_o^{k}) =  MG(S_{rem}^{k'-1},s_r^{k'})$.
	\end{enumerate}
\end{lma}
\begin{proof}
	As $s_r^{k'}$ provides the least marginal gain in $S_{order}$, according to Inequality~\ref{eq:order_rem_2}, any other node in $V\setminus\{{\sf TIR \cup S_{rem}}\}$ cannot be present in $S_{order}\setminus S_{rem}$. Hence, all new seed nodes come from {\sf TIR}. This completes the proof of the first part.
	
The second part of the theorem also holds, since the new seed nodes (i.e., $S_{order}\setminus S_{rem}$) present in {\sf TIR}
do not affect the marginal gain of the old seed nodes (i.e., $S_{rem}$) outside {\sf TIR}. It is because they are at least {\sf 2-Family} away from old seed nodes (Lemma~\ref{lma:rem_mg}).
\end{proof}

Now, we are ready to prove that Inequality~\ref{eq:greedy_req} holds for Case-1 (i.e., $MG(S_{order}^{k-1},s_o^k) \ge MG(S_{order},w)$) and Case-B (i.e., $s_{o}^k = s_{r}^{k'}$).
\begin{lma}\label{lma:subcase2}
	If $MG(S_{order}^{k-1},s_o^k) \ge MG(S_{order},w)$ (Case-1) where $w$ is the top-most node in the priority queue, and $s_o^k = s_r^{k'}$ (Case-B) then $S_{order}$ is the set of seed nodes computed by Greedy on the updated graph, i.e., Inequality~\ref{eq:greedy_req} holds.
\end{lma}
\begin{proof}
We prove this lemma for the nodes present in {\sf TIR} and outside {\sf TIR} separately.
	
First, we will prove that the lemma is true for $u \in V\setminus\{{\sf TIR}\cup S_{order}\}$.
As $S_{rem}^{k'-1} \subseteq S_{order}^{k-1}$, due to Lemma~\ref{lma:submodularity} (i.e., sub-modularity):
\begin{align}\label{eq:sremsord}
		MG(S_{rem}^{k'-1},u) \ge MG(S_{order}^{k-1},u), \quad \forall u \in V\setminus\{{\sf TIR}\cup S_{order}\}
\end{align}
When $s_o^k = s_r^{k'}$, $S_{order}\setminus S_{rem} \in {\sf TIR}$ (Lemma~\ref{lma:seedtir}.1). Hence, $TIR\cup S_{order} = TIR\cup S_{rem}$. From Inequality~\ref{eq:order_rem_2},
\begin{align}\label{sremsrk'u}
		MG(S_{rem}^{k'-1},s_r^{k'}) \ge MG(S_{rem}^{k'-1},u), \quad \forall u  \in V\setminus\{{\sf TIR}\cup S_{order}\}
\end{align}
By combining Lemma~\ref{lma:seedtir}.2, Inequality~\ref{eq:sremsord}, and Inequality~\ref{sremsrk'u}, we get $MG(S_{order}^{k-1},s_o^k) \ge MG(S_{order}^{k-1},u)$, for all $u \in V\setminus \{\text{\sf TIR}\cup S_{order}\}$.

Now, what is left to be proved is that Inequality~\ref{eq:greedy_req} holds for all nodes $u \in \text{\sf TIR}\setminus S_{order}$.
As every such node $u$ is at least {\sf 2-Family} away from $s_r^{k'}$, according to Lemma~\ref{lma:rem_mg},
\begin{align}
MG(S_{order}^{k-1},u) = MG(S_{order}^{k},u), \quad \forall u \in \text{\sf TIR}\setminus S_{order}
\end{align}
Since $w$ is the top-most node in the priority queue, and our assumption is that $MG(S_{order}^{k-1},s_o^k) \ge MG(S_{order},w)$, then
\begin{align}
MG(S_{order}^{k-1},s_o^k) \ge MG(S_{order},w) & \ge MG(S_{order},u)  \nonumber \\
& = MG(S_{order}^{k-1},u) \label{eq:sord_so_u}
\end{align}
From the Inequality~\ref{eq:sord_so_u}, we get $MG(S_{order}^{k-1},s_o^k) \ge MG(S_{order}^{k-1},u)$ for all nodes $u \in \text{\sf TIR}\setminus S_{order}$. This completes the proof.
\end{proof}
Now, we show that for Case-2, i.e., $MG(S_{order}^{k-1},s_o^k)$ $< MG(S_{order},$ $w)$, where $w$ is the top node in the priority queue, by doing iterative seed replacement for a maximum of $k'$ times, we achieve Case-1. Hence, our {\sf N-Family} method
provides the seed set same as the one provided by Greedy on the updated graph. First, we prove that for Case-2, only Case-B (i.e., $s_o^k$ = $s_r^{k'}$) holds, and $w \in {\sf TIR}\setminus S_{order}$.
\begin{lma}\label{lma:iter_seed_last}
	Consider $w$ as the top-most node in the priority queue, and $MG(S_{order}^{k-1},s_o^{k}) < MG(S_{order},w)$ (Case-2). Then,\\
	1. $s_o^{k} = s_{r}^{k'}$ (Case-A)\\ 
	2. $w \in TIR\setminus S_{order}$
\end{lma}
\begin{proof}
We prove both parts of this lemma by contradiction.
	
For the first part, let us assume $s_o^k \ne s_{r}^{k'}$, which means $s_o^k = s_{n}^k$(Case-B). For all nodes $u \in V\setminus{S_{new}}$,

From Lemma~\ref{lma:submodularity}, we get:
\begin{align}\label{eq:snew_1_u}
MG(S_{new}^{k-1},u) \ge MG(S_{new},u).
\end{align}
Since $S_{new} = S_{order}$, and by combining Inequality~\ref{eq:snew_greedy} and Inequality~\ref{eq:snew_1_u}, we get $MG(S_{order}^{k-1},s_o^k) \ge MG(S_{order},u)$. This contradicts the given condition. Hence, $s_o^k = s_r^k$.
	
For the second part of the lemma, let us assume that $w \in V\setminus\{{\sf TIR} \cup S_{order}\}$ (obviously, $w \notin S_{order}$).
From the first part of the lemma, we have $s_{o}^k = s_{r}^{k'}$.

Since $S_{order}\setminus S_{rem} \in {\sf TIR}$, ${\sf TIR} \cup S_{order} = {\sf TIR} \cup S_{rem}$. From Inequality~\ref{eq:order_rem_2}, for $w \in V\setminus{TIR \cup S_{order}}$, we get:
\begin{align}\label{eq:lma11_1}
MG(S_{rem}^{k'-1},s_r^{k'}) \ge MG(S_{rem}^{k'-1},w)
\end{align}
Since $S_{rem}^{k-1} \subset S_{order}$, from Lemma~\ref{lma:submodularity}, we have:
\begin{align}\label{eq:lma11_2}
MG(S_{rem}^{k'-1},w) \ge MG(S_{order},w)
\end{align}
Following Lemma~\ref{lma:seedtir}.2, Inequality~\ref{eq:lma11_1} and Inequality~\ref{eq:lma11_2}, we get that $MG(S_{order}^{k-1},s_o^k)\ge MG(S_{order},w)$, which contradicts our assumption. Hence, $w \in TIR\setminus S_{order}$. This completes the proof.
\end{proof}
Now, in our iterative seed replacement phase, we begin with removing $s_o^k$ ($=s_r^{k'}$) from $S_{order}$; for every node $u\in F_2(s_o^k)\setminus S_{order}$, compute the marginal gain $MG(S_{order}^{k-1},u)$, and update the priority queue.
After updating the priority queue, we compute the new seed node from the updated graph by running Greedy over it. The new seed node computed comes from ${\sf TIR}\setminus S_{order}^{k-1}$, more specifically, it is the top-most node $w$ in the priority queue, as demonstrated below.
\begin{lma}\label{lma:w_iterseed}
If $w$ is the top-most node in the priority queue and $MG(S_{order}^{k-1},s_o^{k}) < MG(S_{order}\cup \{s_o^{k}\},w)$, then the new seed node that replaces $s_o^k$ ($= s_r^{k'}$) is $w$. To prove the lemma, we prove that $MG(S_{order}^{k-1},w) \ge MG(S_{order}^{k-1},u)$ for all nodes $u \in V\setminus S_{order}^{k-1}$.
\end{lma}
\begin{proof}
From Lemma~\ref{lma:iter_seed_last}.1, $s_o^k = s_r^{k'}$ and from Lemma~\ref{lma:iter_seed_last}.2, $w\in {\sf TIR}\setminus S_{order}$. We prove this lemma for the nodes present in {\sf TIR} and outside {\sf TIR} seperately.

First, we shall prove that $MG(S_{order}^{k-1},w) \ge MG(S_{order}^{k-1},u)$ holds for $u \in V\setminus \{TIR \cup S_{order}^{k-1}\}$.
From Lemma~\ref{lma:submodularity} (i.e., sub-modularity), we get:
\begin{align}
MG(S_{order}^{k-1},w) &\ge MG(S_{order}^{k-1}\cup \{s_o^{k'}\},w)  & \nonumber \\
&> MG(S_{order}^{k-1},s_o^{k}) \quad \text{(due to Case-2 condition)}& \label{eq:w_iter1}
\end{align}
Combining Inequality~\ref{eq:w_iter1}, Lemma~\ref{lma:seedtir}.2, and Inequality~\ref{eq:order_rem_2} we get:
\begin{align}\label{eq:w_iter2}
MG(S_{order}^{k-1},w) > MG(S_{rem}^{k'-1},u)
\end{align}
%
According to Lemma~\ref{lma:submodularity} (submodularity), Inequality~\ref{eq:w_iter2} can be written as $MG(S_{order}^{k-1},w) > MG(S_{order}^{k-1},u)$ for all nodes $u \in V\setminus \{TIR \cup S_{order}^{k-1}\}$.
	
What is left to be proved is that $MG(S_{order}^{k-1},w) \ge MG(S_{order}^{k-1},u)$ holds for all $u \in {\sf TIR}\setminus S_{order}^{k-1}$. Since $w$ is the top-most node in the priority queue, we get:
\begin{align}\label{eq:w_iter3}
MG(S_{order}^{k-1}\cup \{s_r^{k'}\}, w) \ge MG(S_{order}^{k-1}\cup \{s_r^{k'}\}, u)
\end{align}
Moreover, $u$ and $w$ are at least {\sf 2-Family} away from $s_r^{k'}$. Thus, according to  Lemma~\ref{lma:rem_mg}, Inequality~\ref{eq:w_iter3} can be written as
\begin{align}
MG(S_{order}^{k-1},w) \ge MG(S_{order}^{k-1},u)
\end{align}
This completes the proof.
\end{proof}

After computing the new seed node, we check if we arrived at Case-1. If so, we terminate the algorithm, and report $S_{order}$ as the set of new seed nodes. Otherwise, we execute this iterative process for maximum of $k'$ times to reach Case-1. The following lemma ensures that the iterative seed replacement for a maximum of $k'$ times leads us to Case-1.
\begin{lma}\label{lma:maxiter}
Iterative seed replacement for a maximum of $k'$ times leads us to Case-1, i.e., $MG(S_{order}^{k-1},s_o^k) \ge MG(S_{order},w)$, $w$ is the top-most node in the priority queue.
\end{lma}
\begin{proof}
According to Lemma~\ref{lma:iter_seed_last}.1, for Case-2, the seed node with the least marginal gain belongs to $S_{rem}$.
Since $|S_{rem}| = k'$, we can perform a maximum of $k'$ replacements. Assume that we executed iterative seed replacement for $k'-1$ times, and we are still at Case-2. According to Lemma~\ref{lma:w_iterseed}, the new seed nodes computed for the past $k'-1$ times came from {\sf TIR}. At this stage, $s_r^1$ is the remaining old seed node outside {\sf TIR}, and $S_{order}^{k-1}$ are the set of seed nodes inside {\sf TIR} (Lemma~\ref{lma:w_iterseed}). Let $x$ be the top-most node in the priority queue; hence, $x \in TIR\setminus S_{order}^{k-1}$ (Lemma~\ref{lma:last_seed}.1), and $s_o^k = s_r^1$ (Lemma~\ref{lma:last_seed}.2). According to Lemma~\ref{lma:w_iterseed}, $x$ would be the new seed node. Hence, we shall prove that $MG(S_{order}^{k-1}, x) \ge MG(S_{order}^{k-1}\cup \{x\}, u)$, for all $u \in V\setminus S_{order}^{k-1}$. Since we prove for all nodes  $u \in V\setminus S_{order}^{k-1}$, it is true for $w$ also. We will prove the lemma for the nodes present in {\sf TIR} and outside {\sf TIR} separately.
	
First, we will prove that $MG(S_{order}^{k-1}, x) > MG(S_{order}^{k-1}\cup \{x\}, u)$ for all nodes $u \in V\setminus TIR$.
Since, $s_r^1$ is at least {\sf 2-Family} away from $x$, according to Lemma~\ref{lma:rem_mg} and our assumption that $x$ is the top-most node in the priority queue, we get:
\begin{align}\label{eq:lma13_1_1}
		MG(S_{order}^{k-1}, x)	= MG(S_{order}^{k-1}\cup \{s_r^1\}, x) > MG(S_{order}^{k-1}, s_r^1)
\end{align}
Since $\{S_{order}^{k-1}\cup\{{s_r^1}\}\}\setminus \{s_r^1\} \in {\sf TIR}$ (Lemma~\ref{lma:seedtir}.1), $s_r^1$ is the only seed node outside {\sf TIR}, and is at least {\sf 2-Family} away from all seed nodes in {\sf TIR}, the Inequality~\ref{eq:lma13_1_2} can be written as
\begin{align}\label{eq:lma13_1_2}
        MG(S_{order}^{k-1}, x) > MG(S_{order}^{k-1}, s_r^1) = \sigma(s_r^1) \ge \sigma(u)
\end{align}
Since the influence spread of a node is always greater than or equal to its marginal gain with respect to any seed set, Inequality~\ref{eq:lma13_1_2} can be written as
\begin{align}
        MG(S_{order}^{k-1}, x) > MG(S_{order}^{k-1}\cup \{x\},u)
\end{align}
What is left to be proved is that $MG(S_{order}^{k-1}, x) \ge MG(S_{order}^{k-1}\cup \{x\}, u)$ for all nodes $u \in {\sf TIR}\setminus S_{order}^{k-1}$.
According to our assumption that $x$ is the top-most node in the priority queue, we get:
\begin{align}\label{eq:lma13_2_1}
		MG(S_{order}^{k-1}\cup \{s_r^1\},x) \ge MG(S_{order}^{k-1}\cup \{s_r^1\},u)
\end{align}
We also have $x \in {\sf TIR}\setminus \{S_{order}^{k-1}\}$ (Lemma~\ref{lma:last_seed}.2). Moreover, $u$ and $x$ are at least {\sf 2-Family} away from $s_r^1$. Following Lemma~\ref{lma:rem_mg}, the Inequality~\ref{eq:lma13_2_1} can be written as
\begin{align}\label{eq:lma13_2_2}
		MG(S_{order}^{k-1},x) \ge MG(S_{order}^{k-1},u)
\end{align}
On the other hand, from Lemma~\ref{lma:submodularity} (submodularity), we get:
\begin{align}\label{eq:lma13_2_3}
		MG(S_{order}^{k-1},u) \ge MG(S_{order}^{k-1}\cup \{x\},u)
\end{align}
Following Inequality~\ref{eq:lma13_2_2} and Inequality~\ref{eq:lma13_2_3}, we get:
$MG(S_{order}^{k-1}, x) > MG(S_{order}^{k-1}\cup \{x\}, u)$, for all nodes $u \in {\sf TIR}\setminus S_{order}^{k-1}$.

After completing $k'$ iterations, $S_{order}^{k-1}\cup \{x\}$ becomes $S_{order}$, $s_o^k = x$, and $MG(S_{order}^{k-1}, x) \ge MG(S_{order}, u)$ for all nodes $u \in V\setminus {S_{order}}^{k-1}$. Hence,  $MG(S_{order}^{k-1}, s_o^k) > MG(S_{order}, w)$, where $w$ is the top-most node in the priority queue. Hence, the lemma.
\end{proof}
\begin{thrm}
The top-$k$ seed nodes reported by our {\sf N-Family} method provides $(1-\frac{1}{e})$ approximation guarantee to the optimal solution,
under the MIA model.
\end{thrm}
\begin{proof}
The top-$k$ seed nodes reported by our {\sf N-Family} method are the same as the top-$k$ seed nodes obtained by running the Greedy on the updated graph
under the MIA model (by following Lemma~\ref{lma:subcase1} and Lemma~\ref{lma:subcase2}). Since, the Greedy algorithm provides the approximation guarantee of $1-\frac{1}{e}$
under the MIA model \cite{CWW10},  our {\sf N-Family} also provides the same approximation guarantee.
\end{proof}